\documentclass{WileyMSP-template}

\usepackage[utf8]{inputenc}
\usepackage{ragged2e}
\usepackage{lmodern}
\usepackage{helvet}
\usepackage{geometry}
\usepackage{amsmath}
\usepackage{amssymb}
\usepackage{amsthm}
\usepackage{siunitx}
\usepackage{sansmath}
\usepackage{mdframed}

\usepackage{graphicx}
\usepackage{grffile}
\usepackage{caption}
\usepackage{subcaption}
\usepackage{tikz}
\usetikzlibrary{external}
\usepackage{pgfplots}
\pgfplotsset{compat=newest}
\usepackage{tikzscale}
\usepackage[percent]{overpic}
\usepackage{epstopdf}
\usepackage{import}
\usepackage{placeins}

\usepackage{booktabs}
\usepackage{supertabular}
\usepackage{array}
\usepackage{multirow}
\usepackage{rotating}
\usepackage{tabularx}
\usepackage[table]{xcolor} 

\usepackage[ruled, vlined, linesnumbered, english]{algorithm2e}
\usepackage{listings}

\usepackage{csquotes}
\usepackage{xifthen}
\usepackage[pdfpagelabels, hidelinks]{hyperref}
\usepackage[acronym, nonumberlist, section]{glossaries}

\usepackage{academicons}
\newcommand{\orcid}[1]{\href{https://orcid.org/#1}{\textcolor[HTML]{A6CE39}{\aiOrcid}}}

\usepgfplotslibrary{fillbetween}

\definecolor{cX1}{HTML}{0072B2}   
\definecolor{cX2}{HTML}{D55E00}   
\definecolor{cBnd}{HTML}{444444}  

\pgfplotsset{
    sx1/.style      = {color=cX1},                       
    sx2/.style      = {color=cX2},                       
    unconstr/.style = {line width=1.0pt,                 
                       dash pattern=on 5pt off 3pt, line cap=round},
    aware/.style    = {line width=1.6pt, solid,          
                       line cap=round, line join=round},
    bnd/.style      = {color=cBnd, line width=0.8pt,     
                       dash pattern=on 1.2pt off 2.4pt, line cap=round},
}

\newtheorem{definition}{Definition}
\newtheorem{theorem}{Theorem}

\begin{document}
\justifying

\pagestyle{fancy}
\rhead{\includegraphics[width=2.5cm]{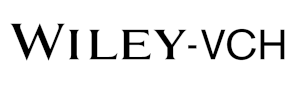}}

\begin{center}
\title{Approximate Feedback Nash Equilibria in Constrained Differential Games via Model Predictive Control with Upper Bound Guarantees}    
\end{center}

\maketitle


\author{Balint Varga$^{*1}$,}
\author{Karl Handwerker$^{1}$,}
\author{Imre Rudas$^2$,}
\author{Peter Galambos$^2$,}



\begin{affiliations}
$^*$ Corresponding Author: Balint Varga, Email Address: \url{balint.varga2@kit.edu} 

$^1$ B. Varga and K. Handwerker are with the Institute of Control Systems, Karlsruhe Institute for Technology, Karlsruhe, D-76131, Germany

$^2$ I. Rudas and P. Galambos are with the Antal Bejczy Center for Intelligent Robotics, Research and Innovation Center of Obuda University, Budapest, Hungary 

B.Varga and K.Handwerker contributed equally to this work.

\end{affiliations}



\begin{abstract}
Physical human--machine interaction and other multi-agent control settings require decision-making policies that adapt online. Differential games provide a principled framework where each agent optimizes an individual objective while anticipating the other's response. The relevant solution concept in many applications is the feedback Nash equilibrium (FNE), which yields time-consistent state-feedback strategies. However, computing the FNE is demanding and becomes intractable when state and input constraints must be enforced, motivating the need for approximate methods. This paper presents a Model Predictive Control (MPC) approach that approximates infinite-horizon FNE trajectories through repeated solution of finite-horizon open-loop games. An auxiliary-game formulation is introduced that selects prediction horizons and terminal costs to approximate the feedback-game optimality conditions. The approach is extended to incorporate hard constraints via a constrained open-loop game formulation. For the unconstrained setting, an analytic upper bound on the state-trajectory deviation between the MPC-induced and FNE trajectories is derived, enabling quantitative performance certification. Numerical examples illustrate the effectiveness of the proposed method compared with baseline approaches from the literature. 
\end{abstract}

\section{Introduction}
\label{sec:introduction}

Physical human--machine interaction (pHMI) and shared-control systems require controllers that adapt online to a partner's actions while enforcing safety constraints and actuator limits. In driving-assistance and haptic guidance, a common design goal is to continuously negotiate control authority and intent through the shared interface \cite{abbink2012haptic,mulder2012sharing,varga_human_2024}. These applications involve real-time interaction between human and machine, where both agents continuously adjust their strategies based on the evolving state of the interaction. This iterative shared control paradigm, often implemented through haptic feedback, demands control algorithms that can respond to changing circumstances within tight computational budgets.

Game-theoretic formulations are a well-established fit for these settings because they explicitly model strategic coupling between agents and the resulting closed-loop interaction \cite{li2019differential,music2020haptic}. The application of differential games to human--machine interaction is particularly well-suited because it treats both agents (human and machine) as players in a noncooperative game. Each player optimizes their own objective while anticipating the other player's response, which naturally captures the interaction in shared control settings. Recent work on human-in-the-loop shared control emphasizes the importance of online adaptation while maintaining safety guarantees \cite{HVROC_2023_Kille,varga2024toward}.

In many interactive control problems, the most appropriate equilibrium concept is a \emph{feedback Nash equilibrium} (FNE) because it produces time-consistent state-feedback policies. However, computing FNEs for general nonlinear differential games is rarely tractable in real time, as it involves solving coupled Hamilton--Jacobi--Bellman equations \cite{basar_dynamic_1999}. Even in the linear--quadratic (LQ) case, feedback and open-loop equilibria can differ substantially, so repeatedly applying open-loop solutions may yield systematic closed-loop mismatch \cite{engwerda_uniqueness_2008}. Existing numerical game solvers can compute local equilibria in nonlinear settings and can incorporate constraints, but their computational cost is often incompatible with the sampling rates required in pHMI and shared-control applications \cite{fridovich-keil_iterative_2020,le_cleach_algames_2020}.

Model predictive control (MPC) is the standard framework for online optimization with explicit constraint handling via repeated finite-horizon re-optimization \cite{b_rawlings_model_2024,mayne2000constrained}. This motivates moving-horizon and game-theoretic MPC schemes that repeatedly solve finite-horizon local games \cite{van_den_broek_moving_1999,saltari2022movinghorizon,williams2018brmpc,ko2020shared}. Nevertheless, solving finite-horizon open-loop games inside MPC in general fails to recover the infinite-horizon FNE of the underlying interaction, and the resulting gap resists a priori quantification.

This paper addresses approximate computation of FNE behavior in constrained differential games using an MPC-based construction that preserves standard feasibility machinery while explicitly targeting feedback-game optimality. The main contributions are:
\begin{itemize}
    \item an \emph{auxiliary-game} MPC formulation that selects the prediction horizon and terminal costs to reduce the mismatch to the infinite-horizon feedback-game optimality conditions, including coupling terms between players;
    \item an extension to hard state and input constraints via constrained open-loop local-game solves; and
    \item for the unconstrained setting, an analytic upper bound on the deviation between the MPC-induced trajectory and the infinite-horizon FNE trajectory, enabling quantitative certification of approximation quality.
\end{itemize}

The remainder of the paper is organized as follows: Section~\ref{sec:related_work} reviews related work; Section~\ref{sec:preliminaries} summarizes differential-game preliminaries; Section~\ref{ch:Approx_FNE} presents the auxiliary-game MPC approximation for constrained games; Section~\ref{sec:error_estimation} derives upper bounds for the unconstrained trajectory deviation; and Section~\ref{sec:numerical_examples} reports numerical results.

\section{Related Work}
\label{sec:related_work}

Differential games are commonly used to represent strategic coupling between agents (e.g., human and machine) in real-time interactive control. Across the literature, methods typically trade off among (i) the equilibrium solution concept (feedback Nash, open-loop Nash, or pragmatic non-equilibrium concepts), (ii) the computational machinery used (PDE/Riccati, iterative optimal control, learning), and (iii) their ability to explicitly enforce hard state/input constraints online.

\subsection{Exact Feedback Nash Equilibria}\label{sec:rw_exact}
For general nonlinear differential games, an exact Feedback Nash Equilibrium (FNE) is characterized by a coupled system of Hamilton--Jacobi--Bellman (HJB) equations. Solving these coupled PDEs is rarely feasible beyond low-dimensional settings due to the curse of dimensionality \cite{basar_dynamic_1999}. In the unconstrained linear--quadratic (LQ) case, the FNE can be computed via coupled Riccati equations, for which specialized numerical algorithms and existence/uniqueness results are available \cite{engwerda_algorithms_2007}. However, for applications involving hard actuation limits and safety constraints (as in pHMI and shared control), Riccati-based solutions do not directly apply, motivating approximate or numerically tractable alternatives.

\subsection{Learning-Based Approximations of Feedback Structure}
A wide range of approaches use approximate dynamic programming and reinforcement learning to reduce the computational burden of coupled HJB equations while preserving feedback dependence.
Within the LQ domain, policy-iteration schemes have been proposed that replace the coupled HJB system with iterative solutions of Lyapunov-like equations \cite{vamvoudakis_multi-player_2011}, and off-policy learning approaches have been shown to recover FNE solutions under unknown dynamics assumptions \cite{li_policy_2019}.
For nonlinear games, value-function approximation (e.g., neural parameterizations) and actor--critic architectures are commonly used to approximate equilibrium strategies \cite{zhou_nonlinear_2024, jiang_approximate_2020}.
A recurring limitation for safety-critical settings is that constraints are frequently handled indirectly (e.g., via penalties or saturation), which complicates hard feasibility guarantees \cite{abouheaf_policy_2019, mu_policy-iteration-based_2021}.

\subsection{Constrained Receding-Horizon Game Formulations}
Hard constraints are best addressed in a receding-horizon framework, where MPC provides standard machinery for online optimization with explicit feasibility considerations \cite{b_rawlings_model_2024}. This has led to \emph{Game-MPC} formulations in which agents repeatedly solve finite-horizon dynamic games and apply the resulting actions in a receding-horizon manner \cite{sani_realtime_2023}. Modern iterative solvers (e.g., based on sequential quadratic approximations) enable online computation of local equilibria in nonlinear settings \cite{le_cleach_algames_2020, fridovich-keil_iterative_2020,di_lamperski_ddp_2019,le_cleach_algames_2022}.

A central methodological gap is that many tractable online schemes compute finite-horizon \emph{open-loop} Nash equilibria (OLNE). Because open-loop and feedback equilibria can differ substantially, repeated OLNE computation does not, in general, recover the infinite-horizon FNE \cite{engwerda_uniqueness_2008}. Recent work has proposed intermediate equilibrium concepts and problem reformulations aimed at narrowing this discrepancy, including moving-horizon equilibrium notions and variational-inequality viewpoints \cite{saltari_nash_2022, benenati_linear-quadratic_2024}. 
In contrast, the approach taken here preserves standard open-loop MPC machinery for constraint handling while introducing an auxiliary-game construction to approximate the infinite-horizon FNE and to provide an analytical bound on the induced trajectory deviation.

\subsection{Research Gap and Contributions}
What is missing in the literature is a receding-horizon game formulation that keeps the constraint-handling of open-loop MPC while explicitly targeting, and approximating, the infinite-horizon feedback Nash equilibrium.
Therefore, this paper develops an MPC-based method that approximates infinite-horizon FNE trajectories using repeated solutions of finite-horizon open-loop games. The key contributions are as follows: (i) we introduce an \emph{auxiliary-game} MPC formulation that selects horizon/terminal design to better match feedback-game optimality conditions, (ii) extend the approach to incorporate state and input constraints via constrained open-loop game solves, and (iii) provide an analytic upper bound on the state-trajectory deviation between the MPC-induced closed loop and the feedback Nash equilibrium in the unconstrained case.

\section{Dynamic Game-Theoretic Preliminaries}\label{sec:preliminaries}
This section summarizes the concepts of differential games used in this paper; standard references are \cite{basar_dynamic_1999,engwerda_lq_2005}, where a comprehensive background can be found. Motivated by pHMI and shared-control scenarios in which agents typically pursue individual objectives \cite{chang_collaborative_2023}, we focus on noncooperative $N$-player differential games.

\begin{definition}[Differential Game] \label{2_DiffaGame}
Let the set of players be $\mathbb{P}=\{1,\dots,N\}$. A (finite- or infinite-horizon) differential game is defined by the tuple
\[
\mathcal{G} \;=\; \big(\mathbb{P},\, \mathbb{R}^n,\, \{ \mathbb{R}^{m_i} \}_{i\in\mathbb{P}},\, f,\, \{J_i\}_{i\in\mathbb{P}}\big),
\]
where the shared state $x(\cdot)\in\mathbb{R}^n$ evolves according to
\[
\dot{x}(t)=f\!\big(x(t),u_1(t),\dots,u_N(t),t\big),\qquad x(0)=x_0,
\]
with control inputs $u_i(\cdot):[0,T]\to\mathbb{R}^{m_i}$ chosen by player $i$. Each player $i\in\mathbb{P}$ seeks to minimize the cost functional
\[
J_i(u_1,\dots,u_N)
\;=\;
h_i(x(T),T)
\;+\;
\int_{0}^{T} r_i\!\big(x(t),u_1(t),\dots,u_N(t),t\big)\,dt,
\]
over a finite horizon $T<\infty$, or over the infinite-horizon case obtained as $T\to\infty$, subject to the system dynamics.
\end{definition}

Available information determines admissible strategies and is formalized via the information set. Two information structures are considered: open-loop and feedback.

\begin{definition}[Information Set and Information Structure]\label{def_InformationSet}
The information set of player $i$ at time $t$ is
\[
\eta_i(t) \in P_{\neg \emptyset}(\{x_0, x(t)\}).
\]
\begin{itemize}
    \item Open-loop (OL): $\eta_i(t)=\{x_0\}$ for all $t\in[0,T]$.
    \item Feedback (FB): $\eta_i(t)=\{x(t)\}$ for all $t\in[0,T]$.
\end{itemize}
\end{definition}

In the OL case, players use only $x_0$ and fix their controls at $t=0$ without further adaptation.

\begin{definition}[Open-Loop Strategy]\label{def_open_loop_strategy}
An open-loop strategy $\gamma_i$ for player $i\in\mathbb{P}$ is
\[
u_i(t)=\gamma_i(x_0,t), \quad \forall t\in[0,T],
\]
where $\gamma_i$ is continuous in $t$ for every $x_0\in\mathbb{R}^n$. The set of such strategies is denoted by $\Gamma^{\mathrm{OL}}$.
\end{definition}
In the FB case, controls adapt to the current state during the game.

\begin{definition}[Feedback Strategy]\label{def_feedback_strategy}
A feedback strategy $\gamma_i$ for player $i\in\mathbb{P}$ is
\[
u_i(t)=\gamma_i(x(t),t), \quad \forall t\in[0,T],
\]
where $\gamma_i$ is continuous in $t$ and globally Lipschitz in $x$. The set of such strategies is denoted by $\Gamma^{\mathrm{FB}}$.
\end{definition}

Linear-quadratic (LQ) differential games form a subclass of differential games with linear dynamics and quadratic cost functionals. Their structure often permits analytical solutions and efficient numerical methods. LQ games are defined as follows.

\begin{definition}[Linear-Quadratic Differential Game \cite{engwerda_lq_2005}]\label{def_LQ_differential_game}
An LQ differential game is defined by the same elements as Definition \ref{2_DiffaGame}, with linear dynamics
\begin{equation} \label{eq:linear_system}
    \dot{x}(t) = A x(t) + \sum_{i=1}^N B_i u_i(t),
\end{equation}
where $x(t)\in\mathbb{R}^n$, $u_i(t)\in\mathbb{R}^{m_i}$, and $A$, $B_i$, $i\in\mathbb{P}$, are matrices of appropriate dimensions. The pairs $(A,B_i)$, $i\in\mathbb{P}$, are stabilizable. In addition, the cost function of player $i$ is quadratic:
\begin{equation}
    J_i = \frac{1}{2} x^\top(t_0+T) Q_{i,T} x(t_0+T)
    + \frac{1}{2} \int_0^T \Big( x^\top(t) Q_i x(t)
    + \sum_{j=1}^N u_j^\top(t) R_{ij} u_j(t) \Big)\, dt,
\end{equation}
where $Q_{i,T}$, $Q_i$, and $R_{ij}$ are symmetric for all $i,j\in\mathbb{P}$ and $R_{ii}>0$.
\end{definition}

LQ problems are more tractable than their nonlinear counterparts \cite{chiu_what_2024}. Nonlinear problems are often approximated by linearizing the dynamics around an operating point or trajectory and solving the resulting LQ problem, which yields fast and sufficiently accurate solutions in several engineering applications.

The remainder of the paper uses these structural properties to construct an MPC-based approximation of infinite-horizon FNE behavior while retaining online tractability and compatibility with constraints.

\section{Approximation of Feedback Nash Equilibria via Constrained MPC}\label{ch:Approx_FNE}

This section presents the main results of our work: a novel approach for approximating infinite-horizon FNE in LQ differential games using a constrained MPC framework. 

We consider infinite-horizon linear--quadratic (LQ) differential games with feedback information structure and seek computationally tractable approximations of the associated feedback Nash equilibrium (FNE). The approach is based on a receding-horizon (MPC) construction: at each time, a finite-horizon \emph{open-loop} game is solved and only the first control action is applied. This yields a state-feedback law through repeated re-optimization. While the presentation below focuses on two players, the formulation extends to $N$ players by augmenting the dynamics with additional inputs and adding one cost functional per player; see, e.g., Chapters 7--8 in \cite{engwerda_lq_2005}.

Section~\ref{Moving_Horizon_Game} summarizes the baseline unconstrained moving-horizon game formulation of \cite{van_den_broek_moving_1999} and highlights the resulting feedback law, and Section~\ref{subsec_analytic_solution_to_the_unconstrianed_mpc} recalls its closed-form analytical solution. Section~\ref{sec_Auxiliary_Game_with_Terminal_Cost} introduces the auxiliary-game extension used in this work, in which local \emph{auxiliary} games include terminal penalties to better approximate the infinite-horizon FNE, and Section~\ref{sec_auxiliary_game_constrained_description} extends this construction to state and input constraints to approximate constrained infinite-horizon FNE trajectories.

\subsection{Model Predictive Control in Dynamic Games}\label{Moving_Horizon_Game}

Model predictive control (MPC) constructs a feedback policy by repeatedly solving a finite-horizon optimization problem and applying only the initial control move. In the dynamic-game setting of~\cite{van_den_broek_moving_1999}, the finite-horizon problem is an \emph{open-loop} game; in the continuous-time limit treated in this section, this local game is re-solved at every instant $t$ and only the instantaneous control value is applied, so that the resulting receding-horizon implementation is a state-feedback policy with static gain $K^{MPC}$, even though each local problem uses open-loop information. Throughout, we refer to the approach of \cite{van_den_broek_moving_1999} as the \emph{unconstrained MPC approach}.

Consider an LQ differential game over an infinite time horizon with feedback information structure. In the MPC context, this game is referred to as the \emph{global game}.

\begin{definition}[Global Game in the unconstrained MPC Solution Concept \cite{van_den_broek_moving_1999}]\label{def_global_game} \leavevmode\\
	Let the state dynamics of an infinite-horizon LQ differential game with feedback information structure be
	\begin{equation}\label{MPC_dynamic_system}
		\dot{x}(t) = Ax(t) + \sum_{i \in \mathbb{P}} B_i u_i(t,x(t)),
	\end{equation}
	where $x(t)\in\mathbb{R}^n$ and $u_i(t,x(t))\in\mathbb{R}^{m_i}$ for $i=1,\dots,N$. The strategy space of player $i$, denoted $\mathcal{U}_i$, is a subset of all maps from $[0,\infty)\times\mathbb{R}^n$ to $\mathbb{R}^{m_i}$ such that for each $(u_1,\dots,u_N)\in\mathcal{U}_1\times\cdots\times\mathcal{U}_N$ there exists a unique state trajectory on $[0,\infty)$ for each initial state. Let the cost functional of player $i$ be
	\begin{equation}
		J_i(u_1,\dots,u_N,x_0)
		= \frac{1}{2}\int_0^\infty \left(x(t)^\top Q_i x(t) + \sum_{j\in\mathbb{P}} u_j(t,x(t))^\top R_{ij}u_j(t,x(t))\right)\,dt,
		\quad i\in\mathbb{P},
	\end{equation}
	where $Q_i = Q_i^\top \succeq 0$ and $R_{ii} = R_{ii}^\top \succ 0$ for all $i\in\mathbb{P}$, $R_{ij} = R_{ij}^\top$ for $i \neq j$, and $x(t)$ is the solution of \eqref{MPC_dynamic_system} with $x(0)=x_0$. Positive definiteness of $R_{ii}$ guarantees that the stationarity conditions can be solved uniquely for $u_i$, while $Q_i \succeq 0$ ensures that the infinite-horizon cost is bounded from below. The infinite-horizon differential game is referred to as the \textit{global game}.
\end{definition}

The unconstrained MPC approach approximates an equilibrium behavior for the global game by repeatedly solving finite-horizon \emph{local games}.

\begin{definition}[Local Game in the unconstrained MPC Solution Concept \cite{van_den_broek_moving_1999}]\label{def_local_game} \leavevmode\\
	Consider, at each time $t_0$, a finite-horizon game on $[t_0,t_0+T]$ for some fixed horizon length $T>0$. It is assumed that the information pattern is open-loop for each local game. The strategy space of player $i$, denoted $\hat{\mathcal{U}}_i$, is a subset of all maps from $[t_0,t_0+T]$ to $\mathbb{R}^{m_i}$ such that for each $(u_1,\dots,u_N)\in\hat{\mathcal{U}}_1\times\cdots\times\hat{\mathcal{U}}_N$ there exists a unique state trajectory on $[t_0,t_0+T]$ for each initial condition. Let the cost functional be
	\begin{equation}
		\hat{J}_i(\hat{u}_1,\dots,\hat{u}_N,x(t_0))
		= \frac{1}{2}\int_{t_0}^{t_0+T}\left(x(t)^\top Q_i x(t) + \sum_{j\in\mathbb{P}} u_j(t)^\top R_{ij}u_j(t)\right)\,dt,
		\quad i\in\mathbb{P},
	\end{equation}
	where $x(t)$ is the solution of \eqref{MPC_dynamic_system} on $[t_0,t_0+T]$ with initial condition $x(t_0)$. This finite-horizon game is referred to as the \textit{local game}.
\end{definition}

\textbf{Remark:} The global and local games differ in their information structure: the global game uses feedback information, whereas each local game is formulated with open-loop information. The receding-horizon implementation applies only the instantaneous value at $t = t_0$ of the local open-loop equilibrium and re-solves the local game from the current state at each subsequent time instant, yielding a closed-loop (state-dependent) control law.

For time-invariant dynamics and costs, it is sufficient to analyze a representative local game on $[0,T]$; the receding-horizon policy is obtained by repeatedly re-solving the same finite-horizon problem with updated initial state.

\begin{theorem}[Time-Invariant unconstrained MPC Solution \cite{van_den_broek_moving_1999}]\label{theo_time_invariant_unconstrained_mpc_solution} \leavevmode\\
	Let $T>0$ and let $x_0\in\mathbb{R}^n$. Let 	$u_i^*(\,\cdot\,;x_0):[0,T]\to\mathbb{R}^{m_i}$, $i\in\mathbb{P}$, be an open-loop Nash equilibrium of the local game of Definition~\ref{def_local_game} on $[0,T]$ with initial state $x_0$, and denote by
	\begin{equation}\label{eq_initial_move_map}
		u_i^{(0)*}(x_0) := u_i^*(0;x_0), \qquad
		u_i^{(0)*}:\mathbb{R}^n\to\mathbb{R}^{m_i},
	\end{equation}
	its value at the initial time of the horizon, regarded as a map from the initial state to the applied control. Then the state-feedback strategies $u_i(x(t)) = u_i^{(0)*}(x(t))$, obtained by evaluating this map at the current state,	form a \textbf{time-invariant unconstrained MPC solution} for the global game with moving horizon length $T$.
\end{theorem}

\textbf{Proof:}
See the proof of Theorem 2.2 in \cite{van_den_broek_moving_1999}.

\textbf{Remark:} Theorem~\ref{theo_time_invariant_unconstrained_mpc_solution} is stated in the continuous-time limit, where the local game is re-solved at every $t$ and the instantaneous control value is applied, yielding the static gain $K^{MPC}$. This is to be distinguished from the sampled implementation used in Section~\ref{subsec_constrained_example}, where the local games are re-solved at discrete sampling instants ($\Delta t = 0.005\,\text{s}$) with a zero-order hold on the applied
input; the two coincide as $\Delta t \to 0$.

The receding-horizon construction is \emph{time consistent} in the sense of \cite{van_den_broek_moving_1999}: at any time along the closed-loop trajectory, the applied control equals the first move of the equilibrium of the newly initialized local game at the current state. Importantly, because each local problem uses open-loop information, the resulting closed-loop law does not, in general, coincide with the FNE of the global feedback game.

\subsection{Analytical Solution to the Unconstrained MPC Approach}\label{subsec_analytic_solution_to_the_unconstrianed_mpc}

The analytical solution of the local open-loop LQ game can be obtained from the associated two-point boundary value problem (TPBVP); see \cite{engwerda_lq_2005}. For the two-player case, the TPBVP is
\begin{equation}\label{eq_TPBVP_unconstrained_MPC}
	\begin{split}
		\begin{bmatrix}
			\dot{x}(t)\\
			\dot{\psi}_1(t)\\
			\dot{\psi}_2(t)
		\end{bmatrix}
		&=-
		\begin{bmatrix}
			-A & S_1 & S_2\\
			Q_1 & A^\top & 0\\
			Q_2 & 0 & A^\top
		\end{bmatrix}
		\begin{bmatrix}
			x(t)\\
			\psi_1(t)\\
			\psi_2(t)
		\end{bmatrix}
		= -M
		\begin{bmatrix}
			x(t)\\
			\psi_1(t)\\
			\psi_2(t)
		\end{bmatrix},
	\end{split}
\end{equation}
with $S_i = B_iR_{ii}^{-1}B_i^\top$, $i\in\mathbb{P}$, and
\begin{equation}
	M=
	\begin{bmatrix}
		-A & S_1 & S_2\\
		Q_1 & A^\top & 0\\
		Q_2 & 0 & A^\top
	\end{bmatrix}.
\end{equation}

A closed-form expression for the solution of the finite-horizon open-loop game can be written in terms of the matrix exponential; see \cite{engwerda_lq_2005}.%
\footnotemark
\footnotetext{The expressions from \cite{engwerda_lq_2005}, p.~340 - p.~343 are stated here for the local-game structure of \cite{van_den_broek_moving_1999}, i.e., without terminal costs.}
Define
\begin{equation}
	H(T)=
	\begin{bmatrix}
		I & 0 & 0
	\end{bmatrix}
	e^{-MT}
	\begin{bmatrix}
		I\\
		0\\
		0
	\end{bmatrix}.
\end{equation}
The two-player local game has an open-loop Nash equilibrium (OLNE) for every initial state $x_0$ if and only if $H(T)$ is invertible \cite{engwerda_lq_2005}, p. 266. Moreover, with
\begin{align}
	H(T-t) &=
	\begin{bmatrix}
		I & 0 & 0
	\end{bmatrix}
	e^{-M(T-t)}
	\begin{bmatrix}
		I\\
		0\\
		0
	\end{bmatrix},\\
	G_1(T-t) &=
	\begin{bmatrix}
		0 & I & 0
	\end{bmatrix}
	e^{-M(T-t)}
	\begin{bmatrix}
		I\\
		0\\
		0
	\end{bmatrix},\\
	G_2(T-t) &=
	\begin{bmatrix}
		0 & 0 & I
	\end{bmatrix}
	e^{-M(T-t)}
	\begin{bmatrix}
		I\\
		0\\
		0
	\end{bmatrix},
\end{align}
the state and the adjoint trajectories of the local game are given by
\begin{equation}\label{eq_state_adjoint_no_terminal}
	x(t) = H(T-t)\,H^{-1}(T)\,x_0,
	\qquad
	\psi_i(t) = G_i(T-t)\,H^{-1}(T)\,x_0,
	\qquad i\in\mathbb{P},
\end{equation}
so that the optimal open-loop controls can be written as (cf.~\cite{engwerda_lq_2005})
\begin{equation}\label{1_13}
	\begin{split}
		u_i^*(t)
		&= -R_{ii}^{-1}B_i^\top \psi_i(t)
		= -R_{ii}^{-1}B_i^\top G_i(T-t)\,H^{-1}(T)\,x_0
		= -R_{ii}^{-1}B_i^\top P_i(t,T)\,x(t),
	\end{split}
\end{equation}
where the last equality follows from \eqref{eq_state_adjoint_no_terminal}, because $H^{-1}(T-t)\,x(t) = H^{-1}(T)\,x_0$, and where
\begin{equation}\label{P_i(t, T)}
	P_i(t,T)=G_i(T-t)H^{-1}(T-t).
\end{equation}
Note that $P_i(t,T)$ maps the \emph{current} state $x(t)$ to the adjoint $\psi_i(t)$, whereas $G_i(T-t)H^{-1}(T)$ maps the \emph{initial} state $x_0$ to it. The receding-horizon feedback law uses the first move ($t=0$), for which both expressions coincide, i.e.,
\begin{equation}\label{P_i(0, T)}
	P_i(0,T)=G_i(T)H^{-1}(T).
\end{equation}

\begin{theorem}[Solution to the unconstrained MPC Problem \cite{van_den_broek_moving_1999}]\label{theo_solution_to_the_unconstrained_MPC_problem} \leavevmode\\
	Let $T>0$. Define
	\begin{equation}
		\begin{bmatrix}
			H(T)\\
			G_1(T)\\
			G_2(T)
		\end{bmatrix}
		= e^{MT}
		\begin{bmatrix}
			I\\
			0\\
			0
		\end{bmatrix},
		\quad
		M=
		\begin{bmatrix}
			-A & S_1 & S_2\\
			Q_1 & A^\top & 0\\
			Q_2 & 0 & A^\top
		\end{bmatrix},
	\end{equation}
	and $S_i=B_iR_{ii}^{-1}B_i^\top$, $i\in\mathbb{P}$. Then there exists a unique MPC solution for the unconstrained, time-invariant linear--quadratic differential game if and only if $H(T)$ is invertible. If this condition holds, the MPC feedback law is
	\begin{equation}\label{feedback_law_mpc_solution}
		u_i^{MPC}(x,T) = -R_{ii}^{-1}B_i^\top G_i(T)H^{-1}(T)\,x, \quad i=1,2.
	\end{equation}
\end{theorem}

\textbf{Proof:}
See the proof of Theorem 3.1 in \cite{van_den_broek_moving_1999}.

The horizon length $T$ directly affects the feedback gains in \eqref{feedback_law_mpc_solution}. As $T$ increases, the corresponding open-loop trajectories generated by the local games converge to the infinite-horizon OLNE trajectories under the assumptions of \cite{van_den_broek_moving_1999}.

\begin{theorem}[Convergence of the unconstrained MPC Solution \cite{van_den_broek_moving_1999}]\label{ConvergenceOfTheMHSolution} \leavevmode\\
	Let $H(T)$ be regular for all $T>0$. Let $P_i(T)=G_i(T)H^{-1}(T)$ and suppose that $\bar{P}_i=\lim_{T\to\infty}P_i(T)$ exists for $i=1,2$. Furthermore, assume that $A-S_1\bar{P}_1-S_2\bar{P}_2$ is stable. Then the open-loop strategies corresponding to the MPC solution obtained in the limit $T\to\infty$ constitute an OLNE for the infinite-horizon game.
\end{theorem}
\textbf{Proof:}
See the proof of Theorem 3.7 in \cite{van_den_broek_moving_1999}.

Although the receding-horizon implementation produces a state-feedback law, it is generated from open-loop local equilibria and therefore does not generally equal the FNE of the global feedback game. In particular, for large $T$, the limiting behavior described above is tied to the infinite-horizon \emph{open-loop} equilibrium, which in general differs from the FNE. This motivates enriching the local problems with additional design degrees of freedom and, crucially for practical settings, incorporating state and input constraints.

\subsection{Approximation of Feedback Nash Equilibria with an Auxiliary Game and Terminal Cost}\label{sec_Auxiliary_Game_with_Terminal_Cost}

The baseline approach of \cite{van_den_broek_moving_1999} is analytically tractable in the unconstrained case and proposes selecting $T$ to reduce the mismatch between the MPC-induced feedback and the global-game feedback optimality conditions. However, in multidimensional settings and in the presence of coupling terms in the cost, optimizing only the horizon length can be too restrictive, as demonstrated in Section \ref{sec:numerical_examples}.

In this section, we consider an extension that augments each local game with terminal
penalties and explicitly incorporates cross terms $R_{ij}$ in the performance index.
Terminal penalties are standard in MPC and have also been analyzed in differential-game MPC
variants \cite{benenati_linear-quadratic_2024}, where they serve their classical role as a
stability and feasibility ingredient for the receding-horizon \emph{open-loop} equilibrium.
Here, in contrast, $Q_{1T}$ and $Q_{2T}$ are free design parameters tuned jointly with the
horizon $T$ to match the optimality conditions of the infinite-horizon \emph{feedback} game,
which, together with explicit state and input constraints, is to the best of our knowledge
not treated systematically in existing work.

The resulting method introduces additional degrees of freedom beyond the horizon length and yields a constrained MPC scheme intended to approximate constrained FNE trajectories via repeated solution of constrained open-loop local games.

The global game is as in Definition \ref{def_global_game}. We now extend the local game in Definition \ref{def_local_game} by adding terminal costs; this extended local problem is referred to as a \textit{local auxiliary game}.

\begin{definition}[Local Auxiliary Game in the LQ MPC Solution Concept]\label{def_local_auxiliary_game_unconstrained} \leavevmode\\
	Consider at each point in time $t_0$ a finite-horizon game on $[t_0,t_0+T]$ for some fixed horizon length $T$, with an open-loop information pattern and terminal-cost parameters $Q_{iT}$ (in contrast to Definition \ref{def_local_game}). The strategy space of player $i$, noted as $\hat{\mathcal{U}}_i$, is defined as the subset of all maps from $[t_0,t_0+T]$ to $\mathbb{R}^{m_i}$ such that for each $(u_1,\dots,u_N)\in\hat{\mathcal{U}}_1\times\cdots\times\hat{\mathcal{U}}_N$ there exists a unique state trajectory on $[t_0,t_0+T]$ for each initial state. Let the cost functional be
	\begin{equation}
		\begin{split}
			\hat{J}_i(\hat{u}_1,\dots,\hat{u}_N,x(t_0))
			&= \frac{1}{2}\int_{t_0}^{t_0+T}\left(x(t)^\top Q_i x(t) + \sum_{j\in\mathbb{P}} u_j(t)^\top R_{ij}u_j(t)\right)\,dt
			+ \frac{1}{2}x^\top(t_0+T)Q_{iT}x(t_0+T), \quad i\in\mathbb{P}.
		\end{split}
	\end{equation}
	where the terminal penalties satisfy $Q_{iT} = Q_{iT}^\top \succeq 0$ for each player $i$, and $x(t)$ is the solution of \eqref{MPC_dynamic_system} on $[t_0,t_0+T]$ with initial condition $x(t_0)$. This game is referred to as the \textit{local auxiliary game}.

\end{definition}

The terminal penalties modify the finite-horizon open-loop solution and induce a dependence of the local Riccati solutions on $(T,Q_{1T},Q_{2T})$. In particular, the state and adjoint trajectories can be expressed as
\begin{align}
	x(t) &= H(T-t,Q_{1T},Q_{2T})\,H^{-1}(T,Q_{1T},Q_{2T})\,x_0,\\
	\psi_1(t) &= G_1(T-t,Q_{1T},Q_{2T})\,H^{-1}(T,Q_{1T},Q_{2T})\,x_0,\\
	\psi_2(t) &= G_2(T-t,Q_{1T},Q_{2T})\,H^{-1}(T,Q_{1T},Q_{2T})\,x_0,
\end{align}
with
\begin{align}
	H(T-t,Q_{1T},Q_{2T}) &= \begin{bmatrix}\label{appendix_H_matrix}
		I & 0 & 0
	\end{bmatrix}e^{-M(T-t)}
	\begin{bmatrix}
		I\\
		Q_{1T}\\
		Q_{2T}
	\end{bmatrix},\\
	G_1(T-t,Q_{1T},Q_{2T}) &=
	\begin{bmatrix}
		0 & I & 0
	\end{bmatrix}e^{-M(T-t)}
	\begin{bmatrix}
		I\\
		Q_{1T}\\
		Q_{2T}
	\end{bmatrix},\\
	G_2(T-t,Q_{1T},Q_{2T}) &=
	\begin{bmatrix}
		0 & 0 & I
	\end{bmatrix}e^{-M(T-t)}
	\begin{bmatrix}
		I\\
		Q_{1T}\\
		Q_{2T}
	\end{bmatrix}.
\end{align}

\textbf{Remark:} For $T>0$, the local auxiliary game admits a unique open-loop Nash equilibrium for every initial state if and only if $H(T,Q_{1T},Q_{2T})$ is invertible; the regularity condition of Theorem~\ref{theo_solution_to_the_unconstrained_MPC_problem} thus carries over unchanged, i.e.\ the terminal-penalty degrees of freedom do not weaken it.

The corresponding local-game Riccati solutions can be written as
\begin{equation}\label{eq_P_i(T, Q_{1T}, Q_{2T})}
	\begin{split}
		P_i(T,Q_{1T},Q_{2T}) = G_i(T,Q_{1T},Q_{2T})\,H^{-1}(T,Q_{1T},Q_{2T}).
	\end{split}
\end{equation}

In contrast to \eqref{P_i(t, T)}, where the first argument of $P_i$ denotes the current time, we write $P_i$ here directly for the value used by the receding-horizon law, i.e.\ at $t=0$; its arguments are therefore the horizon length and the terminal penalties.

Using \eqref{eq_P_i(T, Q_{1T}, Q_{2T})}, we define an optimization problem that tunes both the horizon and the terminal penalties. The solution of the local open-loop auxiliary game is denoted by the subscript $\textit{LG}$ in the following. Substituting $P_i^{LG}(T,Q_{1T},Q_{2T})$ into the coupled algebraic Riccati equations of the infinite-horizon feedback game yields a residual that is generally nonzero; we therefore choose $(T,Q_{1T},Q_{2T})$ to minimize this residual. In contrast to \cite{van_den_broek_moving_1999}, the cross terms $R_{ij}$ are included in the residual through the matrices $S_{ij}$ below.

\begin{definition}[Optimization Problem with Prediction Horizon and Terminal Cost]\label{def_NovelOptimizationApproach} \leavevmode\\
	Let $P_i^{LG}(T,Q_{1T},Q_{2T})$ be the solution to the local game in Definition \ref{def_local_auxiliary_game_unconstrained} on the time horizon $[0,T]$ and $\varphi_i$ the residual matrix obtained by substituting $P_i^{LG}$ into the coupled algebraic Riccati equation for player $i$:
	\begin{equation}
		\begin{split}
			\varphi_i
			&= -Q_i - P_i^{LG}A - A^\top P_i^{LG}
			+ \sum_{j=1}^{2} P_i^{LG} S_j P_j^{LG}
			+ \sum_{\substack{j=1\\ j\neq i}}^{2} P_j^{LG} S_j P_i^{LG}
			- \sum_{\substack{j=1\\ j\neq i}}^{2} P_j^{LG} S_{ij} P_j^{LG},
			\quad i\in\mathbb{P},
		\end{split}
	\end{equation}
	where
	\begin{align}
		S_j &= B_jR_{jj}^{-1}B_j^\top, \quad j\in\mathbb{P},\\
		S_{ij} &= B_jR_{jj}^{-1}R_{ij}R_{jj}^{-1}B_j^\top, \quad i,j\in\mathbb{P},\ i\neq j,
	\end{align}
	and $\|\varphi_i\|$ denotes the Frobenius norm of $\varphi_i$. The optimization problem is defined as
	\begin{equation}\label{eq_auxiliary_optimization_problem}
		\begin{split}
			T^*,\,Q_{1T}^*,\,Q_{2T}^*
			&= \underset{T,\,Q_{1T},\,Q_{2T}}{\text{argmin}}
			\left(\sum_{i=1}^{2}\|\varphi_i\|\right)\\
			\text{subject to}\quad
			& T \in [T_{\min},\,T_{\max}],\quad 0<T_{\min}\leq T_{\max}<\infty,\\
			& Q_{iT} = Q_{iT}^\top \succeq 0, \quad i\in\mathbb{P},\\
			& H(T,Q_{1T},Q_{2T}) \ \text{invertible}.
		\end{split}
	\end{equation}
\end{definition}

\textbf{Remark:} The constraints in \eqref{eq_auxiliary_optimization_problem} restrict the search to the parameters for which the local auxiliary game of Definition \ref{def_local_auxiliary_game_unconstrained} is well posed: the terminal penalties are symmetric positive semidefinite as required there, and the horizon is confined to a compact interval on which the regularity condition $\det H(T,Q_{1T},Q_{2T})\neq 0$ holds, so that $P_i^{LG}$ in \eqref{eq_P_i(T, Q_{1T}, Q_{2T})} is defined.

\textbf{Remark:} The objective in \eqref{eq_auxiliary_optimization_problem} penalizes only the residual of the coupled algebraic Riccati equations and contains no explicit stability requirement. Consequently, the closed-loop matrix $A^{MPC}$ induced by $(T^*,Q_{1T}^*,Q_{2T}^*)$ is not guaranteed to be Hurwitz and has to be verified a posteriori for the returned minimizer. This verification is also a prerequisite for assumption 2 of Theorem \ref{Theorem_proof_upper_bound}; whenever it fails, the error bound derived there does not apply. The problem is furthermore nonconvex, so that only local minimizers are obtained in general; this is discussed further in Section \ref{sec:numerical_examples}.

\subsection{Problem with State Constraints}\label{sec_auxiliary_game_constrained_description}

We now extend the auxiliary-game approach to incorporate state and input constraints. The goal is to obtain an approximation of \emph{constrained} infinite-horizon FNE trajectories by solving a constrained MPC problem, i.e., by repeatedly solving constrained local open-loop games with terminal penalties.

\begin{definition}[Global Game with State Constraints in the MPC Solution Concept]\label{GlobalGame_constrained} \leavevmode\\
	Let the state dynamics of an infinite-horizon LQ differential game with feedback information structure and state constraints be
	\begin{equation}\label{MPC_dynamic_system_constrained}
		\begin{split}
			\dot{x}(t) &= Ax(t) + \sum_{i\in\mathbb{P}} B_i u_i(t,x(t)),
			\quad x(t)\in\mathcal{X},\quad u_i(t,x(t))\in U_i,\quad i=1,\dots,N.
		\end{split}
	\end{equation}
	with permissible state and control input sets $\mathcal{X}$ and $U_i$.

	The strategy space of player $i$, denoted $\mathcal{U}_i$, is defined as the subset of all maps from $[0,\infty)\times\mathcal{X}$ to $U_i$ such that for each $(u_1,\dots,u_N)\in\mathcal{U}_1\times\cdots\times\mathcal{U}_N$ there exists a unique state trajectory on $[0,\infty)$ for each initial state. Let the cost functional be
	\begin{equation}
		J_i(u_1,\dots,u_N,x_0)
		= \frac{1}{2}\int_0^\infty \left(x(t)^\top Q_i x(t) + \sum_{j\in\mathbb{P}} u_j(t,x(t))^\top R_{ij}u_j(t,x(t))\right)\,dt,
		\quad i\in\mathbb{P},
	\end{equation}
	where $x(t)$ is the solution of \eqref{MPC_dynamic_system_constrained} with $x(0)=x_0$. The resulting infinite-horizon differential game is referred to as the \textit{constrained global game}.
\end{definition}

\begin{definition}[Local Auxiliary Game with Constraints in the MPC Solution Concept]\label{LocalGame_constrained} \leavevmode\\
	Consider at each time $t_0$ a finite-horizon game on $[t_0,t_0+T]$ with fixed horizon length $T>0$. It is assumed that the information pattern is open-loop for each local game. The strategy space of player $i$, noted as $\hat{\mathcal{U}}_i$, is defined as the subset of all maps from $[t_0,t_0+T]$ to $U_i$ such that for each $(u_1,\dots,u_N)\in\hat{\mathcal{U}}_1\times\cdots\times\hat{\mathcal{U}}_N$ there exists a unique state trajectory $x(t)$ satisfying the system dynamics and the constraints $x(t)\in\mathcal{X}$ on $[t_0,t_0+T]$ for each initial state. Let the cost functional be
	\begin{align} \nonumber
		\hat{J}_i(\hat{u}_1,\dots,\hat{u}_N,x(t_0))
		= &\frac{1}{2}\int_{t_0}^{t_0+T}\left(x(t)^\top Q_i x(t) + \sum_{j\in\mathbb{P}} u_j(t)^\top R_{ij}u_j(t)\right)\,dt \\
		&+ \frac{1}{2}x^\top(t_0+T)Q_{iT}x(t_0+T),\quad i\in\mathbb{P},
	\end{align}
	where $x(t)$ is the solution of \eqref{MPC_dynamic_system_constrained} on $[t_0,t_0+T]$ with initial condition $x(t_0)$, subject to $x(t)\in\mathcal{X}$ and $u_i(t)\in U_i$, and the terminal penalties $Q_{iT}$ are as in Definition~\ref{def_local_auxiliary_game_unconstrained}. This game is referred to as the \textit{constrained local auxiliary game}.

\end{definition}
The constrained MPC trajectory is generated by repeatedly solving the constrained local auxiliary game and applying the first control action, thereby enforcing $x(t)\in\mathcal{X}$ and $u_i(t)\in U_i$ by construction. Computing the constrained local open-loop equilibria efficiently is therefore central. Existing approaches include variational-inequality formulations for constrained open-loop games \cite{benenati_linear-quadratic_2024} and iterated best-response methods \cite{williams_best_2018}. Practical applications of such MPC schemes can be found in the literature on multi-agent control and multi-robot systems, however, with strong assumptions on the structure of the game (e.g., decoupled costs) or without systematic analysis of the approximation quality of the resulting trajectories to constrained infinite-horizon FNE trajectories, which is handled in the next section.

\section{Upper Bounds for the Trajectory Error Between FNE and MPC Solutions} \label{sec:error_estimation}

In the previous section, a method for approximating FNE based on an MPC approach was introduced. It remains an open question whether the maximum error between state trajectories can be estimated using only the approximated solution and the reference solution from the coupled Riccati equation. This approach would eliminate the need to perform a simulation-based, a posteriori analysis for every individual scenario. The following section aims at deriving an upper bound for the trajectory error estimates of the MPC and FNE solution for the unconstrained case.

The literature on trajectory-error estimation in differential games highlights three relevant directions: Lyapunov-based bounds, matrix-exponential norm estimates, and $\varepsilon$-Nash concepts.
A dynamic neural-network identification method with Lyapunov analysis for $N$-player nonzero-sum games with unknown nonlinear dynamics is developed in \cite{johnson_approximate_2015}, obtaining bounds primarily on input trajectory errors. 
In \cite{frihauf_nash_2012}, extremum seeking and Lyapunov arguments are used to ensure bounded input deviations, but explicit bounds for state errors are not provided. 
In \cite{cappello_approximate_2020}, $\varepsilon$-Nash equilibria are analyzed, which guarantee that no player can improve their cost by more than $\varepsilon$ through unilateral deviation, but this does not directly translate to explicit state trajectory error bounds.
In contrast, in \cite{varga_upper_2024}, an explicit upper bound on state-trajectory deviations is derived via a norm-based distance metric between games, which is therefore most aligned with estimating maximum state error in our setting.

Following \cite{varga_upper_2024}, we use matrix-exponential bounds to estimate trajectory deviations in infinite-horizon LQ feedback games. The resulting closed-loop dynamics, governed by Riccati-based FNE gains, are:
\begin{equation}\label{err_approx_system_dynamics_fne}
\dot{x}(t)=A^{FNE}x(t),
\end{equation}
with
\begin{equation}
A^{FNE}=A-\sum_{i\in\mathbb{P}}B_i R_{ii}^{-1}B_i^\top P_i = A+\sum_{i\in\mathbb{P}}B_iK_i^{FNE}, \qquad K_i^{FNE}:=-R_{ii}^{-1}B_i^\top P_i,
\end{equation}
with $u_i^{FNE}(t)=K_i^{FNE}x^{FNE}(t)$ the FNE feedback law, and the corresponding trajectory is
\begin{equation}
x^{FNE}(t)=e^{A^{FNE}t}x_0.
\end{equation}
The solution of the MPC approach is based on the optimization problem from Definition \ref{def_NovelOptimizationApproach}, where the control law is defined as:
\begin{equation}
	u_i^{MPC}(t) = K_i^{MPC} x(t)
\end{equation}
The feedback law leads to closed system matrices
\begin{equation} \label{err_approx_system_dynamics_mpc}
	\dot{x}^{MPC}(t) = A^{MPC} x^{MPC}(t),
\end{equation}
where $A^{MPC}=A+\sum_{i\in\mathbb{P}}B_iK_i^{MPC}$ is the closed-loop matrix obtained by substituting $u_i^{MPC}(t)=K_i^{MPC}x(t)$ into \eqref{MPC_dynamic_system}. The solution of \eqref{err_approx_system_dynamics_mpc} is
\begin{equation}
	x^{MPC}(t) = e^{A^{MPC} t} x_0 
\end{equation}

\begin{theorem}[Upper Bound for the Trajectory Error in MPC-based Feedback Nash Games] \label{Theorem_proof_upper_bound}
	Let an LQ game and an LQ local game be given as defined in Definitions \ref{2_DiffaGame} and \ref{def_local_auxiliary_game_unconstrained}. 
	Let $x^{MPC}(t)$ denote the state trajectory of the game based on the auxiliary game solution, and $x^{FNE}(t)$ represent the state trajectory corresponding to the FNE of the game. Assume the following conditions:
	\begin{enumerate}
		\item The initial conditions are identical, i.e., $x^{MPC}(0) = x^{FNE}(0) = x_0$.
		\item The system matrices $A^{MPC}$ and $A^{FNE}$ are Hurwitz matrices, i.e., their spectra lie in $\{z \in \mathbb{C} : \operatorname{Re}(z)~<~0\}$.
		\item The constants $M^{MPC}, M^{FNE} > 0$ and $\lambda^{MPC}, \lambda^{FNE} > 0$ satisfy $\|e^{A^{MPC} t}\| \leq M^{MPC} e^{-\lambda^{MPC} t}$ and $\|e^{A^{FNE} t}\| \leq M^{FNE} e^{-\lambda^{FNE} t}$ for $t \geq 0$.
		\item The control gain matrices differ by $\Delta K_i = K^{MPC}_i - K^{FNE}_{i}$ for each player $i \in \mathbb{P}$.
	\end{enumerate}

	Then, the error $e(t) = x^{MPC}(t) - x^{FNE}(t)$ between the state trajectories is bounded for all $t \geq 0$ by
	\begin{equation}
		\|e(t)\| \leq M^{MPC} M^{FNE} \|x_0\| \left(\sum_{i \in \mathbb{P}} \|B_i \Delta K_i\|\right) \cdot
		\begin{cases}
			\dfrac{e^{-\lambda^{FNE} t} - e^{-\lambda^{MPC} t}}{\lambda^{MPC} - \lambda^{FNE}}, & \lambda^{MPC} \neq \lambda^{FNE},\\[2mm]
			t\,e^{-\lambda^{MPC} t}, & \lambda^{MPC} = \lambda^{FNE}.
		\end{cases}
	\end{equation}
\end{theorem}

\textbf{Remark:} Both closed-loop systems \eqref{err_approx_system_dynamics_fne} and \eqref{err_approx_system_dynamics_mpc} are time invariant, so the choice $t_0 = 0$ in assumption 1 is without loss of generality. For an arbitrary initial time $t_0$ with $x^{MPC}(t_0) = x^{FNE}(t_0) = x_0$, the statement holds on $[t_0,t_1]$ with $t$ replaced by $t-t_0$ on both sides of the inequality.

\textbf{Proof:}
The error dynamics are determined by the following equation from the differential equations in \eqref{err_approx_system_dynamics_fne} and \eqref{err_approx_system_dynamics_mpc}:

\begin{align} \label{err_approx_error_dynamics_base}
	\dot{e}(t) = \dot{x}^{MPC}(t) - \dot{x}^{FNE}(t) = A^{MPC} x^{MPC}(t) - \left( A x^{FNE}(t) + \sum_{i \in \mathbb{P}} B_i u_i^{FNE}(t) \right)
\end{align}

The error dynamics are next expressed as a function of the error $e(t)$ itself. Substituting $u_i^{FNE}(t)=K_i^{FNE}x^{FNE}(t)$ and $A^{FNE}=A+\sum_{i\in\mathbb{P}}B_iK_i^{FNE}$ into \eqref{err_approx_error_dynamics_base} gives

\begin{align} \label{error_approx_2_11}
	\dot{e}(t) = A^{MPC} x^{MPC}(t) - A^{FNE} x^{FNE}(t).
\end{align}

Adding and subtracting $A^{MPC}x^{FNE}(t)$ and using $A^{MPC}=A+\sum_{i\in\mathbb{P}}B_iK_i^{MPC}$ yields

\begin{align} \label{error_approx_2_12}
	\dot{e}(t) &= A^{MPC}\big(x^{MPC}(t) - x^{FNE}(t)\big) + \big(A^{MPC}-A^{FNE}\big) x^{FNE}(t) \nonumber \\
	&= A^{MPC} e(t) + \sum_{i \in \mathbb{P}} B_i \left( K_i^{MPC} - K_i^{FNE} \right) x^{FNE}(t).
\end{align}

With $\Delta K_i = K_i^{MPC}-K_i^{FNE}$ as in the theorem statement, \eqref{error_approx_2_12} becomes

\begin{align} \label{error_approx_2_13}
	\dot{e}(t) &= A^{MPC} e(t) + \sum_{i \in \mathbb{P}} B_i \Delta K_i x^{FNE}(t)
\end{align}

Equation \eqref{error_approx_2_13} is a linear inhomogeneous differential equation in $e(t)$, driven by the known FNE trajectory $x^{FNE}(t)$, with $B_i\in\mathbb{R}^{n\times m_i}$ and $\Delta K_i\in\mathbb{R}^{m_i\times n}$. Its solution is determined by the method of variation of constants (see \cite{mironchenko_input--state_2023}):

\begin{align}
	e(t) &= e^{A^{MPC}t} e(0)
	+ \int_0^t e^{A^{MPC}(t- \tau)} \sum_{i \in \mathbb{P}} \left( B_i \Delta K_i x^{FNE}(\tau)\right) d\tau 
\end{align}

To estimate the norm of the error, the triangle inequality is applied. This allows for the separate consideration of the norm of the first and second terms, where the first term corresponds to the homogeneous solution and the second term to the particular solution of the differential equation

\begin{align} \label{error_approx_2_15}
	\|e(t)\| &\leq \|e^{A^{MPC}t} e(0)\| + \int_0^t \|e^{A^{MPC}(t-\tau)}\| \cdot \sum_{i \in \mathbb{P}} \|B_i \Delta K_i\| \cdot \|e^{A^{FNE}\tau} \| \cdot \|x_0\| \, d\tau.
\end{align}

Assuming identical initial conditions for the FNE and MPC trajectories $x^{FNE}(t_0) = x^{MPC}(t_0) = x_0$ and thus $e(0) = x^{MPC}(0) - x^{FNE}(0) = 0 $, \eqref{error_approx_2_15} simplifies to

\begin{align} \label{error_approx_2_16}
	\|e(t)\| &\leq \int_0^t \|e^{A^{MPC}(t-\tau)}\| \cdot \sum_{i \in \mathbb{P}} \|B_i \Delta K_i\| \cdot \|e^{A^{FNE}\tau}\| \cdot \|x_0\| \, d\tau.
\end{align}

The norm of the matrix exponential functions in \eqref{error_approx_2_16} can be bounded by the inequality

\begin{align} \label{error_approx_upper_bound_expm}
	\| e^{At} \| &\leq M e^{- \lambda t}
\end{align}
where $M \in \mathbb{R}^+$ and $\lambda \in \mathbb{R}^+$ \cite{mironchenko_input--state_2023}.
With \eqref{error_approx_upper_bound_expm}, the upper bound on the error is estimated as

\begin{align}
	\|e(t)\| &\leq  \int_0^t M^{MPC} e^{-\lambda^{MPC}(t-\tau)} \cdot \sum_{i \in \mathbb{P}} \|B_i \Delta K_i\| \cdot M^{FNE} e^{-\lambda^{FNE}\tau} \|x_0\| \, d\tau.
\end{align}

To simplify the mathematical expression, terms independent of the integration variable are factored out of the integral.
\begin{align} \label{error_approx_2_19}
	\|e(t)\| &\leq M^{MPC} M^{FNE} e^{-\lambda^{MPC} t } \cdot \sum_{i \in \mathbb{P}} \|B_i \Delta K_i\| \cdot \|x_0\| \ \cdot  \int_0^t   e^{-(\lambda^{FNE} -\lambda^{MPC})\tau} d\tau.
\end{align}

For $\lambda^{MPC} \neq \lambda^{FNE}$, the integral in \eqref{error_approx_2_19} evaluates to $\big(1-e^{-(\lambda^{FNE}-\lambda^{MPC})t}\big)/(\lambda^{FNE}-\lambda^{MPC})$; this quantity is nonnegative for every $t\geq0$, regardless of the sign of $\lambda^{FNE}-\lambda^{MPC}$, since it equals the integral of the nonnegative integrand $e^{-(\lambda^{FNE}-\lambda^{MPC})\tau}$ over $[0,t]$ and must therefore \emph{not} be replaced by its absolute value. Multiplying by the leading $e^{-\lambda^{MPC}t}$ factor and simplifying $e^{-\lambda^{MPC}t}e^{-(\lambda^{FNE}-\lambda^{MPC})t}=e^{-\lambda^{FNE}t}$ yields

\begin{align} \label{error_approx_2_20}
	\|e(t)\| &\leq M^{MPC} M^{FNE} \|x_0\| \left(\sum_{i \in \mathbb{P}} \|B_i \Delta K_i\|\right) \cdot \frac{e^{-\lambda^{FNE} t} - e^{-\lambda^{MPC} t}}{\lambda^{MPC} - \lambda^{FNE}}.
\end{align}
For $\lambda^{MPC} = \lambda^{FNE}$, the integral in \eqref{error_approx_2_19} evaluates to $t$ instead, giving
\begin{align} \label{error_approx_2_20b}
	\|e(t)\| &\leq M^{MPC} M^{FNE} \|x_0\| \left(\sum_{i \in \mathbb{P}} \|B_i \Delta K_i\|\right) \cdot t\, e^{-\lambda^{MPC} t}. \qed
\end{align}
\textbf{Note}: The estimation of the matrix exponential function is further discussed in the Appendix.

The upper bound \eqref{error_approx_2_20} provides a theoretical guarantee on the maximum deviation between the state trajectories of the MPC-based solution and the FNE solution, based on the system properties and the difference in control gains. It highlights how the error depends on the stability properties of both systems (through $M$ and $\lambda$) and the gain mismatch $\Delta K$.

\section{Numerical Examples} \label{sec:numerical_examples}
In order to demonstrate the effectiveness of the proposed method for approximating infinite-horizon FNE using an MPC-based approach, three numerical examples are presented in this section. The first example focuses on a system with two state variables and asymmetric cost functions, the second example extends the analysis to a more complex system with four state variables, and the third example illustrates the constrained formulation of Definition~\ref{LocalGame_constrained}. The results of the simulations are compared to the FNE solution obtained from the coupled algebraic Riccati equations and to the solution according to \cite{van_den_broek_moving_1999}, which serves as a baseline for comparison. The performance of the proposed method is evaluated based on the maximum and the accumulated deviation from the FNE solution.

Throughout this section, the deviation of a simulated state trajectory $x(t)$ from the FNE trajectory $x^{FNE}(t)$ is quantified by two scalar measures. The maximum error
\begin{equation}\label{eq_def_eps_max}
	\varepsilon_{\text{max}} := \max_{t \in [0,t_f]} \left\| x(t) - x^{FNE}(t) \right\|
\end{equation}
denotes the peak deviation over the simulation interval $[0,t_f]$, and the accumulated error
\begin{equation}\label{eq_def_eps_sum}
	\varepsilon_{\text{sum}} := \int_0^{t_f} \left\| x(t) - x^{FNE}(t) \right\| \, dt
\end{equation}
denotes the deviation accumulated over the entire trajectory, where $\|\cdot\|$ is the Euclidean norm of the full state-error vector. The superscripts $(\cdot)^{MPC}$ and $(\cdot)^{Van}$ indicate whether $x(t)$ is the trajectory of the proposed auxiliary-game MPC approach or of the method from \cite{van_den_broek_moving_1999}.

\subsection{2 Degree-of-Freedom Example}\label{subsubsec_example_2}
    A system with two state variables is analyzed, in which also $R_{ij}$ in the cost functional of the first player are considered, providing higher weight to the interaction between the players' inputs.
    
    \begin{align}
        \dot{x} &= \begin{bmatrix} -1 & 2 \\ 0 & -1.5 \end{bmatrix} x + \begin{bmatrix} 0.5 \\ 0 \end{bmatrix} u_1 + \begin{bmatrix} 1 \\ 1 \end{bmatrix} u_2. 
    \end{align}

    The penalties for the state variables are not symmetrical and in the first player's cost function, the cross-term $R_{12}$ has a larger value than $R_{11}$.
    \begin{align}
        J_1 &= \int^{T}_0 x^\top Q_1 x + u_1^2 + u_2^\top R_{12} u_2  \, dt \\
        J_2 &= \int^{T}_0 x^\top Q_2 x + u_2^2 \, dt 
    \end{align}   
    where the parameters are
    \begin{equation}
        Q_1 = \begin{bmatrix} 4 & 0 \\ 0 & 1  \end{bmatrix}, \quad Q_2 = \begin{bmatrix} 3.21 & 0 \\ 0 & 0.64  \end{bmatrix}, \quad R_{11} = R_{22} = 1, \quad R_{12} = 4.2.
    \end{equation}
    Figure \ref{fig:problem_03_state} illustrates the trajectories of the system states that show the difference between the FNE solution, the solution according to \cite{van_den_broek_moving_1999} and the optimized MPC solution with terminal costs. Figure \ref{fig:problem_03_input} shows the control inputs of the two players.
    
    \begin{figure}[!t]
        \centering
        \includegraphics{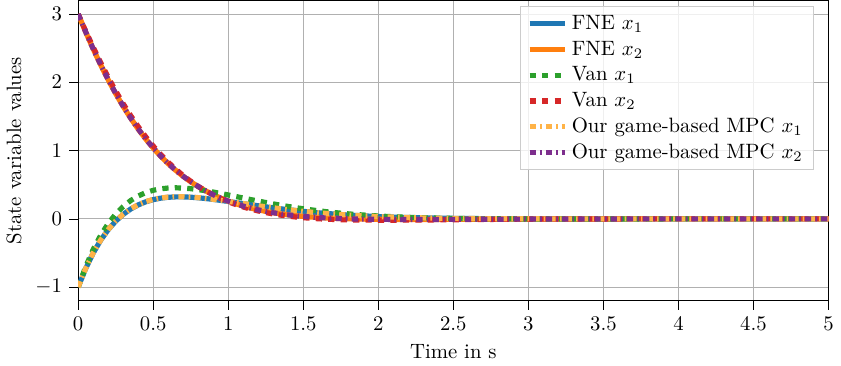}
        \caption{Comparison of exact FNE, van den Broek~\cite{van_den_broek_moving_1999}, and auxiliary-game MPC state trajectories for the two-DoF example system. Solid lines represent FNE, dashed lines represent van den Broek~\cite{van_den_broek_moving_1999}, and dotted lines represent the proposed MPC approach.}
        \label{fig:problem_03_state}    
    \end{figure}

	\begin{figure}[!t]
		\centering
		\includegraphics{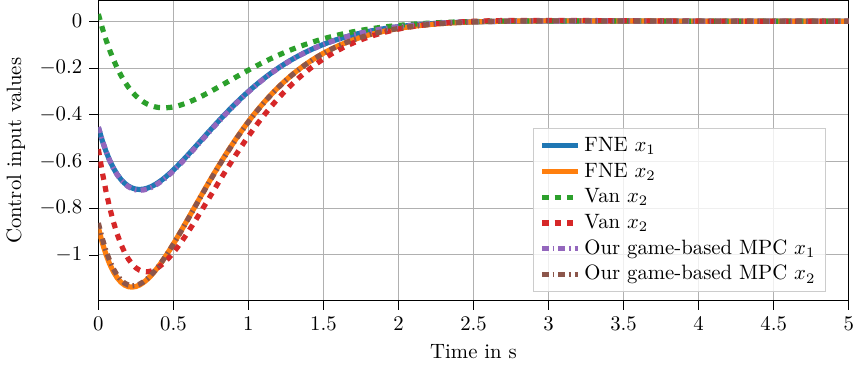}
		\caption{Comparison of the control inputs of the exact FNE, van den Broek~\cite{van_den_broek_moving_1999} and auxiliary-game MPC input trajectories for the two-DoF example system. Line styles follow Figure~\ref{fig:problem_03_state}.}
		\label{fig:problem_03_input}
	\end{figure}
    
    While the solution based on \cite{van_den_broek_moving_1999} shows a maximum deviation of $\varepsilon_{\text{max}}^{\text{Van}} = 0.1397$ from the FNE solution, our novel optimization with terminal costs reduces the error to $\varepsilon_{\text{max}}^{\text{MPC}} = 0.01$. The improvement is more visible in the control inputs: The solution according to \cite{van_den_broek_moving_1999} shows a larger deviation from the FNE solution, while the MPC solution is comparatively close to the FNE solution. Including the coupling terms $R_{ij}$ in the cost functions, the proposed method demonstrates the ability to produce a solution that converges closely to the FNE.

\subsection{4 Degree-of-Freedom Example}\label{subsubsec_example_3}

    This example evaluates the robustness and scalability of the proposed method by considering asymmetric cost functions in combination with increased model complexity. In contrast to the previous example, the system comprises four state variables and two-dimensional control inputs for each player. The analysis focuses on the trajectory errors relative to the FNE solution to assess the effectiveness of the auxiliary-game approach for higher-dimensional  systems.
    
    Consider the following dynamic system with four state variables
    \begin{align}
        \dot{x} &= \begin{bmatrix} -8 & -6 & 1 & 0 \\ 1 & 0 & 2 & 1 \\ 0 & -2 & 0 & 1 \\ 0 & 1 & 0 & -1 \end{bmatrix} x + \begin{bmatrix} 0 & 1 \\ 0 & 0 \\ 0 & 0 \\ 1 & 0 \end{bmatrix} u_1 + \begin{bmatrix} 0 & 0 \\ 0 & 0 \\ 1 & 0 \\ 0 & 1 \end{bmatrix} u_2  
    \end{align}
    and for the local game the asymmetric cost functions
    \begin{align}
        J_1 &= \int^{T}_0 x^\top Q_1 x + u_1^\top R_{11} u_1 + u_2^\top R_{12} u_2  \, dt \\
        J_2 &= \int^{T}_0 x^\top Q_2 x + u_1^\top R_{21} u_1 +  u_2^\top R_{22} u_2  \, dt 
    \end{align}
    in which parameters are defined as:
    $$
        Q_1 = \begin{bmatrix} 1 & 0 & 0 & 0 \\ 0 & 0.4 & 0 & 0 \\ 0 & 0 & 2 & 0 \\ 0 & 0 & 0 & 1 \end{bmatrix}, \quad Q_2 = \begin{bmatrix} 1 & 0 & 0 & 0 \\ 0 & \frac{2}{3} & 0 & 0 \\ 0 & 0 & 1 & 0 \\ 0 & 0 & 0 & 2 \end{bmatrix} 
    $$
    $$
        R_{11} = \begin{bmatrix} 1 & 0 \\ 0 & 1 \end{bmatrix}, \quad R_{12} = \begin{bmatrix} 0.5 & 0 \\ 0 & 0.5 \end{bmatrix}, \quad R_{21} = \begin{bmatrix} 1 & 0 \\ 0 & 1 \end{bmatrix}, \quad  R_{22} = \begin{bmatrix} 1 & 0 \\ 0 & 1 \end{bmatrix}.
    $$

    Figure \ref{fig:problem_04_state} illustrates the state trajectories of the four state variables, while Figure \ref{fig:problem_04_input1} and Figure \ref{fig:problem_04_input2} show the control input trajectories of the first and second player.
    For the proposed method, the maximum deviation from the FNE trajectory is    $\varepsilon_{\text{max}}^{\text{MPC}} = 0.0271$ and the accumulated error is    $\varepsilon_{\text{sum}}^{\text{MPC}} = 0.0178$, compared with    $\varepsilon_{\text{max}}^{\text{Van}} = 0.0512$ and $\varepsilon_{\text{sum}}^{\text{Van}} = 0.1328$ for the method of \cite{van_den_broek_moving_1999}. The auxiliary-game formulation thus halves the peak deviation and reduces the accumulated error by roughly a factor of seven, confirming that the approach retains its accuracy for higher-dimensional dynamics with coupled cost terms.

    The results show that the method presented is also able to reduce the error between the state trajectories and achieve a better approximation of the FNE solution even with more complex system dynamics.

    \begin{figure}[!t]
        \centering
        \includegraphics{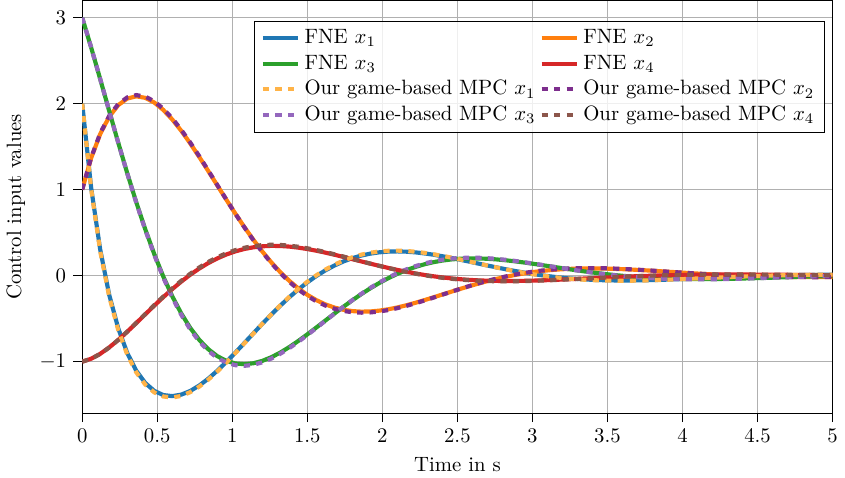}
        \caption{Comparison of state trajectories of the exact FNE and auxiliary-game MPC state trajectories for the four-DoF example system. Solid lines represent FNE and dotted lines represent the proposed MPC approach.}
        \label{fig:problem_04_state}
    \end{figure}
    \begin{figure}[!t]
        \centering
        \includegraphics{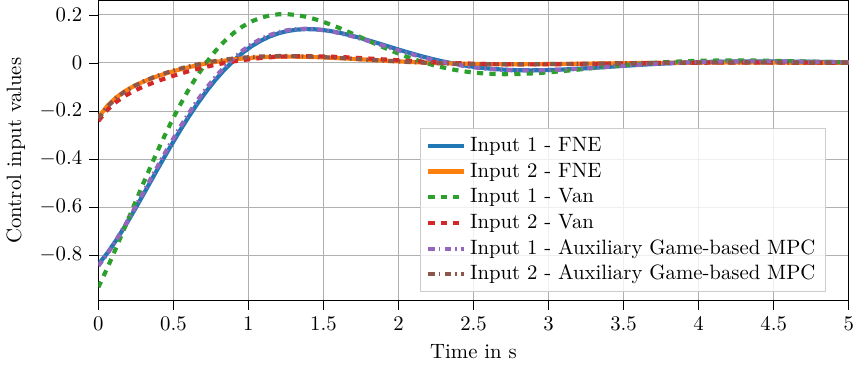}
        \caption{Comparison of the first player's control inputs of the exact FNE, van den Broek~\cite{van_den_broek_moving_1999}, and auxiliary-game MPC input trajectories for Player~1 in the four-DoF example system. Solid lines represent FNE, dashed lines represent van den Broek~\cite{van_den_broek_moving_1999}, and dotted lines represent the proposed MPC approach.}
        \label{fig:problem_04_input1}
    \end{figure}
    \begin{figure}[!t]
        \centering
		\includegraphics{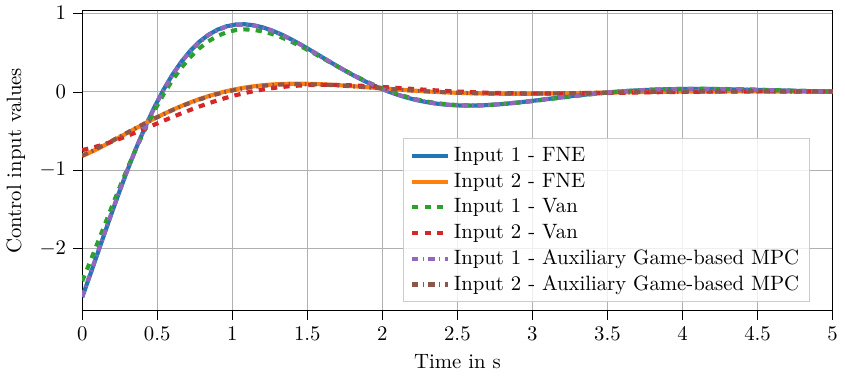}
        \caption{Comparison of teh second player's exact FNE, van den Broek~\cite{van_den_broek_moving_1999}, and auxiliary-game MPC input trajectories for Player~2 in the four-DoF example system. Solid lines represent FNE, dashed lines represent van den Broek~\cite{van_den_broek_moving_1999}, and dotted lines represent the proposed MPC approach.}
        \label{fig:problem_04_input2}
    \end{figure}

\subsection{Constrained Differential Game Example}\label{subsec_constrained_example}
To demonstrate that the constrained formulation of Section~\ref{sec_auxiliary_game_constrained_description} is functional, an underdamped oscillator system is analyzed with an additional state constraint. The system dynamics are given by
\begin{align}
\dot{x}(t) &= A x(t) + B_1 u_1(t) + B_2 u_2(t),
\end{align}
with
\begin{align}
A &=
\begin{bmatrix} 0 & 1 \\ -\omega^2 & -2\zeta\omega \end{bmatrix} =
\begin{bmatrix} 0 & 1 \\ -9 & -0.36 \end{bmatrix}, \\
B_1 &= B_2 =
\begin{bmatrix} 0 \\ 1 \end{bmatrix},
\end{align}
where $\omega=3.0$ and $\zeta=0.06$. The cost function parameters are defined as
\begin{align}
Q_1 &= Q_2 =
\begin{bmatrix} 100 & 0 \\ 0 & 1 \end{bmatrix}, \\
R_{11} &= R_{22} =
\begin{bmatrix} 1 \end{bmatrix}, \\
R_{12} &= R_{21} =
\begin{bmatrix} 0.5 \end{bmatrix}.
\end{align}
The initial condition is
\begin{align}
x_0 &=
\begin{bmatrix} 2.25 \\ 1.0 \end{bmatrix},
\end{align}
and the simulation horizon is $t_f = 4.0\text{ s}$. A state constraint is imposed on the first state variable:
\begin{equation}
x_1(t) \ge -0.15 \quad \forall t \ge 0,
\end{equation}
which represents a safe lower bound on the position variable. A second state constraint is imposed on the second state variable:
\begin{equation}
x_2(t) \ge -2.75 \quad \forall t \ge 0,
\end{equation}
which represents a safe lower bound on the velocity variable. These constraints are enforced in each local auxiliary game according to Definition~\ref{LocalGame_constrained}. The constrained local open-loop games are solved numerically at each sampling instant with a dense plotting resolution of $\Delta t = 0.005\text{ s}$, and only the first control action is applied.

Figure~\ref{fig:state_comparison_constrained_vs_naive} shows the resulting closed-loop state trajectories. The constrained MPC trajectory (blue, solid) satisfies the state constraints $x_1 \ge -0.15$ and $x_2 \ge -2.75$ at all times, while the unconstrained planner (orange, dashed) violates this bound during the initial transient. This confirms that the proposed auxiliary-game MPC scheme retains the constraint-handling capability of standard MPC while approximating equilibrium behavior.

Figure~\ref{fig:input_comparison_constrained_vs_naive} illustrates the control inputs for Player~1. The constraint-aware planner respects the actuator bounds while maintaining feasibility of the state trajectory. In contrast, the unconstrained planner demands aggressive control actions that would be required to achieve the unconstrained FNE but are incompatible with the safety constraints.

\textbf{Remark:} This example serves as a qualitative proof-of-concept to demonstrate that the constrained formulation is functional and enforces hard state constraints. A systematic quantitative analysis of the approximation quality with respect to the constrained infinite-horizon FNE (analogous to Sections~\ref{subsubsec_example_2} and~\ref{subsubsec_example_3}) is beyond the scope of this work and is left for future research. Nevertheless, these results demonstrate that the proposed method successfully enforces hard state constraints while maintaining the game-theoretic structure of the multi-agent interaction.

\begin{figure}[!t]
    \centering
    \includegraphics{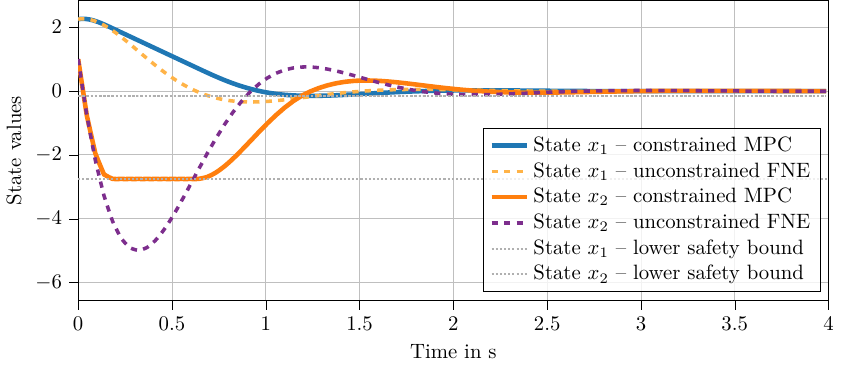}
    \caption{Comparison of state trajectories between constrained MPC and unconstrained FNE on the saturated real plant. Solid lines represent the constrained MPC, dashed lines represent the unconstrained (naive) FNE, and the dotted line indicates the lower safety bound on $x_1$ and $x_2$.}
    \label{fig:state_comparison_constrained_vs_naive}
\end{figure}

\begin{figure}[!t]
    \centering
    \includegraphics{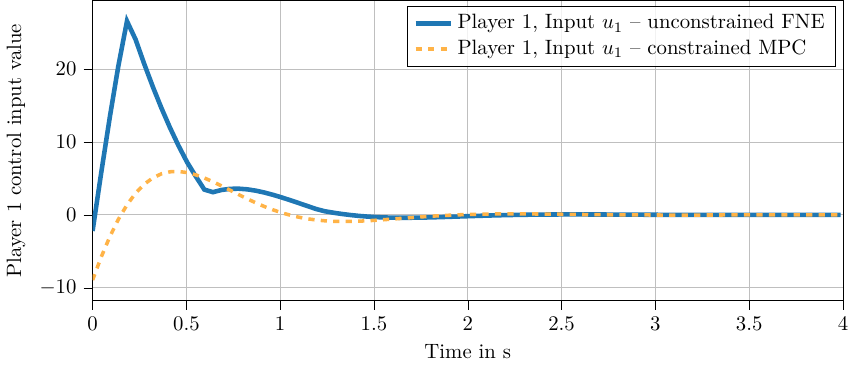}
    \caption{Comparison of control input trajectories for Player~1 between constraint-aware and unconstrained planners. Solid lines represent the constraint-aware planner, dashed lines represent the unconstrained planner.}
    \label{fig:input_comparison_constrained_vs_naive}
\end{figure}

\subsection{Upper Bound Estimations of the Trajectory Error}
In this section, the two unconstrained examples of Subsections~\ref{subsubsec_example_2} and~\ref{subsubsec_example_3} are revisited in order to assess the practical applicability and the accuracy of the error estimate derived in Theorem~\ref{Theorem_proof_upper_bound}. For each system, the estimated upper bound is compared with the simulated trajectory error between the MPC and the FNE solution, and the remaining gap between the two is used to identify the strengths and the limitations of the approach. In both figures of this subsection, the plotted quantity is the Euclidean norm $\|\varepsilon(t)\| = \|x^{MPC}(t) - x^{FNE}(t)\|$ of the full state-error vector.

For the two-DoF example with asymmetric cost functions (Subsection~\ref{subsubsec_example_2}), the estimated bound attains a maximum of $0.0155$, whereas the maximum simulated error is $\varepsilon_{\text{max}}^{\text{MPC}} = 0.0100$, as shown in Figure~\ref{fig:problem_02_state_ERR}. For the four-DoF example (Figure~\ref{fig:problem_04_state_ERR}), the estimated bound attains a maximum of $0.371$ against a maximum simulated error of $\varepsilon_{\text{max}}^{\text{MPC}} = 0.0271$. In both cases the estimate is a valid upper bound over the entire trajectory, and in both cases the deviation between the MPC and the FNE trajectories is certified to remain small in absolute terms. The margin between bound and simulation is tighter for the two-DoF system and larger for the four-DoF system, which reflects the norm-based construction of \eqref{error_approx_2_20b}: the individual gain mismatches enter through $\sum_{i \in \mathbb{P}} \|B_i \Delta K_i\|$ and both closed-loop matrix exponentials are bounded separately by $M e^{-\lambda t}$. Since these steps discard directional information, the margin naturally widens with the state dimension and the number of coupled inputs. The bound is therefore best understood as an \emph{a priori} certificate that is available without simulating the closed loop, rather than as a sharp prediction of the realized error.

The error estimate relies on norm bounds for the matrix exponential, computed using eigenvalues and eigenvectors. The bound employed in this work is
\begin{align}
    \| e^{At} \| &\leq M e^{-\lambda t}.
\end{align}
Alternative constructions, including Lyapunov-based bounds or logarithmic-norm approaches, may yield different -- and potentially more numerically robust -- estimates.

Since the literature on trajectory-error approximation between MPC and FNE solutions remains limited, the proposed method is intended as a conservative upper bound and a basis for future work on improving accuracy and stability, especially for more complex systems.

    \begin{figure}[!h]
        \centering
		\includegraphics{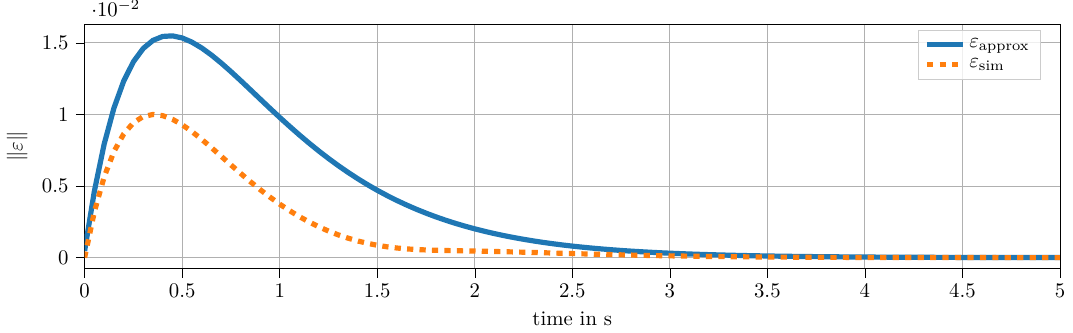}
        \caption{Approximated and simulated trajectory error for the two-DoF example (Section~\ref{subsubsec_example_2}). The estimated upper bound (solid) exceeds the simulated error (dashed).}
        \label{fig:problem_02_state_ERR}
    \end{figure}

    \begin{figure}[!h]
        \centering
        \includegraphics{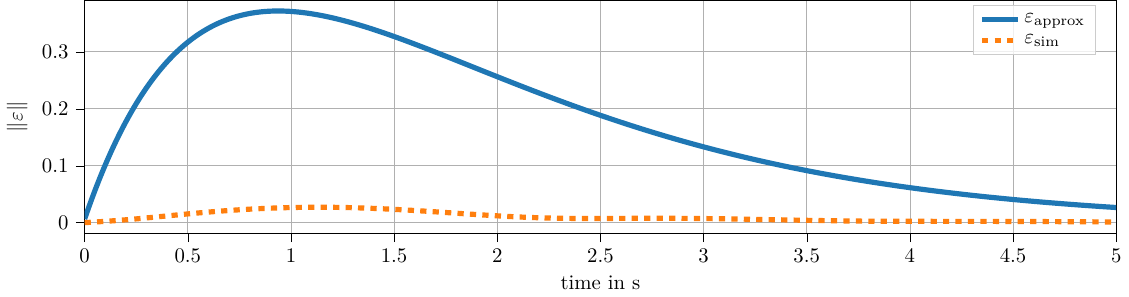}
		\caption{Approximated and simulated trajectory error for the four-DoF example (Section~\ref{subsubsec_example_3}). The estimated upper bound (solid) exceeds the simulated error (dashed).}
        \label{fig:problem_04_state_ERR}
    \end{figure}

\subsection{Discussions}

The numerical examples demonstrate that the auxiliary-game MPC formulation consistently reduces the trajectory deviation from the FNE compared to the baseline horizon-only optimization of \cite{van_den_broek_moving_1999}. In the first example, the maximum state error was substantially reduced, while in the second example -- involving a higher-dimensional system with four states and asymmetric coupling terms -- both the maximum error and the cumulative error over the trajectory were reduced. These results suggest that the additional degrees of freedom introduced by the terminal cost matrices are particularly effective when the cost structure includes significant cross-coupling between players, as captured by the $R_{ij}$ terms that are absent from the baseline formulation.

The upper bound derived in Theorem~\ref{Theorem_proof_upper_bound} was validated against simulated trajectory errors in both examples. The estimated bound exceeded the simulated error in both examples, indicating that the bound is practically useful and not excessively conservative in the cases considered. The tightness of the bound depends on the quality of the matrix exponential estimate, which in this work was obtained via eigenvalue decomposition. Alternative estimation techniques, such as those based on logarithmic norms or Lyapunov equations, may yield tighter bounds in cases where the eigenvalues of the closed-loop system matrices are closely spaced or where the eigenvector matrix is ill-conditioned.

Some limitations should be acknowledged. First, the auxiliary-game optimization problem in Definition~\ref{def_NovelOptimizationApproach} is nonconvex, and the quality of the solution depends on the initialization of the optimizer. No global optimality guarantee is provided for the choice of $T^*$, $Q_{1T}^*$, and $Q_{2T}^*$. Second, the constrained extension presented in Section~\ref{ch:Approx_FNE} is formulated at the problem level but not yet validated numerically in this work; its practical performance remains to be investigated. Third, the upper bound in Theorem~\ref{Theorem_proof_upper_bound} applies only to the unconstrained case. Extending the error bound to constrained settings -- where the closed-loop dynamics are state-dependent and potentially nonsmooth -- is a nontrivial open problem. Finally, the current formulation is restricted to LQ games with time-invariant parameters. Extending the approach to time-varying or nonlinear dynamics, for instance via successive linearization within the MPC loop, is a natural direction for future work.

Despite these limitations, the method offers a practical advantage over policy iteration: it does not require an initial stabilizing feedback gain, and the local open-loop games can be solved using standard optimal control tools compatible with existing MPC infrastructure.

\section{Conclusion and Outlook} \label{sec_summary_approximation_of_fne}

This paper presented an MPC-based method for approximating infinite-horizon feedback Nash equilibria in LQ differential games. The approach extends the baseline moving-horizon game formulation of \cite{van_den_broek_moving_1999} by introducing terminal cost matrices as additional optimization variables in each local open-loop game, thereby increasing the degrees of freedom available to reduce the mismatch with the global feedback-game optimality conditions. Cross-coupling terms between players were explicitly incorporated into the residual minimization, addressing a limitation of prior work that considered only the prediction horizon as a tunable parameter.

The formulation was further extended to accommodate state and input constraints by embedding the auxiliary-game structure within a constrained MPC framework, where each local game is solved as a constrained open-loop Nash problem. For the unconstrained setting, an analytic upper bound on the state-trajectory deviation between the MPC-induced closed loop and the exact FNE was derived using matrix exponential norm estimates. The bound depends on the feedback gain mismatch, the stability margins of both closed-loop systems, and the initial condition, and was shown to be practically informative in the numerical examples considered.

Future work will focus on three directions: (i)~numerical validation of the constrained formulation on benchmark systems relevant to shared control and physical human--machine interaction, see~\cite{handwerker2026iterated}, (ii)~investigation of tighter matrix exponential bounds to improve the conservatism of the error estimate, and (iii)~extension of the approach to nonlinear differential games via iterative linearization within the receding-horizon loop.

\newpage
\medskip
\textbf{Data Availability Statement} \par
The data that support the findings of this study are available from the corresponding author upon request.

\medskip
\textbf{Conflict of Interest} \par
The authors declare no conflict of interest.

\medskip
\textbf{Funding} \par
This work is related to the MedLaBotX project (2024-1.2.3-HU-RIZONT-00069).
Project 2024-1.2.3-HU-RIZONT-00069 has been implemented with support provided by the Ministry of Culture and Innovation of Hungary from the National Research, Development, and Innovation Fund, financed under the 2024-1.2.3-HU-RIZONT funding scheme.

\medskip

%

\newpage
\section*{Appendix}

\subsection*{Estimation of the Matrix Exponential Function}
The proof of Theorem \ref{Theorem_proof_upper_bound} intentionally leaves the choice of estimation method for the matrix exponential open, as various approaches for estimating the upper bounds of  $e^{A^{MPC}}$ and $e^{A^{FNE}}$ are available and well-documented in the literature. In the following, one example of such an estimation is demonstrated, while references to additional possible approaches are provided for further research.
    
    \begin{align}
        \| e^{At} \| &\leq M e^{- \lambda t}.
    \end{align}
    where the constants $M \in \mathbb{R}^+$ and $\lambda \in \mathbb{R}^+$ must be chosen appropriately. Throughout, let
    \begin{equation}\label{eq_spectral_abscissa}
        \alpha(A) := \max_i \operatorname{Re}(\lambda_i)
    \end{equation}
    denote the spectral abscissa of $A$, where $\lambda_1,\dots,\lambda_n$ are the eigenvalues of $A$. Since $A$ is assumed to be Hurwitz, $\alpha(A)<0$.

    \paragraph{Diagonalizable case.}
    A possible estimate is based on the eigenvalues of the matrix $A$ over the decomposition of $A$ into its eigenvectors and the representation
    \begin{align}
        e^{At} &= P e^{\Lambda t} P^{-1}, \quad A = P \Lambda P^{-1}
    \end{align}
    where $\Lambda = \text{diag}(\lambda_1, \lambda_2, \dots, \lambda_n)$ is the diagonal matrix of the eigenvalues of $A$ and $P$ is the corresponding eigenvector matrix. For the norm of $e^{At}$, using submultiplicativity and norm estimation, the following result is obtained:
    \begin{align}
        \| e^{At} \| &= \| P e^{\Lambda t} P^{-1} \| \\
        &\leq \| P \| \cdot \| e^{\Lambda t} \| \cdot \| P^{-1} \|.
    \end{align}
    The exponential function of the diagonal matrix $\Lambda$ is applied component-wise to the eigenvalues:
    \begin{align}
        e^{\Lambda t} &= \text{diag}(e^{\lambda_1 t}, e^{\lambda_2 t}, \dots, e^{\lambda_n t}),
    \end{align}
    so that, in the spectral norm,
    \begin{align}
        \| e^{\Lambda t} \| &\leq e^{\alpha(A) t}, \qquad t\geq 0.
    \end{align}
    This leads to the upper bound
    \begin{align}\label{eq_appendix_diagonalizable_bound}
        \| e^{At} \| &\leq \| P \| \cdot \| P^{-1} \| \cdot e^{\alpha(A) t},
    \end{align}
    where $ \| P \| $ and $ \| P^{-1} \| $ are the constant norms of the eigenvector matrix and its inverse. Hence \eqref{eq_appendix_diagonalizable_bound} is of the required form with $M = \|P\|\cdot\|P^{-1}\| = \kappa(P) \geq 1$ and $\lambda = -\alpha(A) > 0$: the decay rate attains the spectral abscissa, and the constant $M$ is the condition number of the eigenvector matrix. Accordingly, the bound degrades when $P$ is ill-conditioned, i.e., when $A$ is close to a defective matrix.

    \paragraph{Non-diagonalizable case.}
    If the matrix $A$ is not diagonalizable, the Jordan normal form is used instead:
    \begin{equation}
        A = P J P^{-1}, \qquad e^{At} = P e^{Jt} P^{-1},
    \end{equation}
    where $J$ is the Jordan normal form of $A$. In this case $e^{Jt}$ is block diagonal, and the block associated with an eigenvalue $\lambda_i$ of a Jordan block of size $\nu_i$ contains, in addition to $e^{\lambda_i t}$, polynomial terms of degree up to $\nu_i-1$. The estimate \eqref{eq_appendix_diagonalizable_bound} is therefore not directly applicable, since the decay is no longer purely exponential. The required form $\|e^{At}\|\leq Me^{-\lambda t}$ is nevertheless recovered by selecting a decay rate that is strictly slower than the spectral abscissa: for every $\lambda$ with
    \begin{equation}\label{eq_appendix_jordan_rate}
        0 < \lambda < -\alpha(A),
    \end{equation}
    there exists a constant $M \geq 1$ such that $\|e^{At}\| \leq M e^{-\lambda t}$ for all $t\geq 0$. This follows because each entry of $e^{\lambda t}e^{Jt}$ is a finite sum of terms $t^{k}e^{(\lambda_i + \lambda)t}$ with $k \leq \nu_i - 1$ and $\operatorname{Re}(\lambda_i) + \lambda < 0$ by \eqref{eq_appendix_jordan_rate}, and every such term is bounded on $[0,\infty)$; $M$ may then be taken as $\kappa(P)$ times the resulting bound on $\|e^{\lambda t}e^{Jt}\|$. The price of a defective $A$ is thus a larger constant $M$ and a decay rate $\lambda$ that no longer attains $-\alpha(A)$, which makes the bound of Theorem \ref{Theorem_proof_upper_bound} more conservative but leaves its form and validity unchanged.

	Note that a similar proof is presented in \cite{kagstrom_bounds_1977}, but it does not specifically address differential games, and focuses on the applicability of the logarithmic norm for upper bounds of the matrix exponential.


\begin{thebibliography}{99}
\providecommand{\url}[1]{\texttt{#1}}
\providecommand{\urlprefix}{URL }

\bibitem{abbink2012haptic}
D.~A. Abbink, M.~Mulder, E.~R. Boer,
\newblock Haptic shared control: smoothly shifting control authority?,
\newblock \emph{Cognition, Technology \& Work} \textbf{2012}, \emph{14}, 1 19.

\bibitem{mulder2012sharing}
M.~Mulder, D.~A. Abbink, E.~R. Boer,
\newblock Sharing control with haptics: seamless driver support from manual to
  automatic control,
\newblock \emph{Human Factors} \textbf{2012}, \emph{54}, 5 786.

\bibitem{varga_human_2024}
B.~Varga, F.~Flemisch, S.~Hohmann,
\newblock Human in the loop,
\newblock \emph{at -- Automatisierungstechnik} \textbf{2024}, \emph{72}, 12
  1109.

\bibitem{li2019differential}
Y.~Li, G.~Carboni, F.~Gonzalez, D.~Campolo, E.~Burdet,
\newblock Differential game theory for versatile physical interaction,
\newblock \emph{Nature Machine Intelligence} \textbf{2019}, \emph{1}, 1 36.

\bibitem{music2020haptic}
S.~Musi{\'c}, S.~Hirche,
\newblock Haptic shared control for human-robot collaboration: A
  game-theoretical approach,
\newblock In \emph{IFAC-PapersOnLine}, volume~53. \textbf{2020} 10216--10222.

\bibitem{HVROC_2023_Kille}
{S. Kille et al.},
\newblock Human-variability-respecting optimal control for physical
  human-machine interaction,
\newblock In \emph{2024 33rd IEEE International Conference on Robot and Human
  Interactive Communication (ROMAN)}. \textbf{2024} 1595--1602.

\bibitem{varga2024toward}
B.~Varga,
\newblock Toward adaptive cooperation: Model-based shared control using
  lq-differential games,
\newblock \emph{Acta Polytechnica Hungarica} \textbf{2024}, \emph{21}, 10 439
  – 456.

\bibitem{basar_dynamic_1999}
T.~Başar, G.~J. Olsder,
\newblock \emph{Dynamic noncooperative game theory},
\newblock Number~23 in Classics in applied mathematics. SIAM, Philadelphia, 2nd
  ed edition, \textbf{1999}.

\bibitem{engwerda_uniqueness_2008}
J.~Engwerda,
\newblock Uniqueness conditions for the affine open-loop linear quadratic
  differential game,
\newblock \emph{Automatica} \textbf{2008}, \emph{44}, 2 504.

\bibitem{fridovich-keil_iterative_2020}
D.~Fridovich-Keil, V.~Rubies-Royo, C.~J. Tomlin,
\newblock An {Iterative} {Quadratic} {Method} for {General}-{Sum}
  {Differential} {Games} with {Feedback} {Linearizable} {Dynamics},
\newblock In \emph{2020 {IEEE} {International} {Conference} on {Robotics} and
  {Automation} ({ICRA})}. IEEE, Paris, France,
\newblock ISBN 978-1-72817-395-5, \textbf{2020} 2216--2222,
\newblock \urlprefix\url{https://ieeexplore.ieee.org/document/9196517/}.

\bibitem{le_cleach_algames_2020}
S.~Le~Cleac'h, M.~Schwager, Z.~Manchester,
\newblock {ALGAMES}: {A} {Fast} {Solver} for {Constrained} {Dynamic} {Games},
\newblock In \emph{Robotics: {Science} and {Systems} {XVI}}. Robotics: Science
  and Systems Foundation,
\newblock ISBN 978-0-9923747-6-1, \textbf{2020}
  \urlprefix\url{http://www.roboticsproceedings.org/rss16/p091.pdf}.

\bibitem{b_rawlings_model_2024}
J.~B.~Rawlings, D.~Q.~Mayne, M.~M.~Diehl,
\newblock \emph{Model {Predictive} {Control}: {Theory}, {Computation}, and
  {Design} (2nd {Edition})},
\newblock Nob Hill Publishing, LLC, Santa Barbara, California, USA, 2nd edition
  edition, \textbf{2024}.

\bibitem{mayne2000constrained}
D.~Q. Mayne, J.~B. Rawlings, C.~V. Rao, P.~O.~M. Scokaert,
\newblock Constrained model predictive control: Stability and optimality,
\newblock \emph{Automatica} \textbf{2000}, \emph{36}, 6 789.

\bibitem{van_den_broek_moving_1999}
W.~{van den Broek},
\newblock Moving horizon control in dynamic games,
\newblock \emph{Journal of Economic Dynamics and Control} \textbf{2002},
  \emph{26}, 6 937.

\bibitem{saltari2022movinghorizon}
E.~Saltari, W.~Semmler, G.~Di~Bartolomeo,
\newblock A nash equilibrium for differential games with moving-horizon
  strategies,
\newblock \emph{Computational Economics} \textbf{2022}, \emph{60} 1041.

\bibitem{williams2018brmpc}
G.~Williams, B.~Goldfain, P.~Drews, J.~M. Rehg, E.~A. Theodorou,
\newblock Best response model predictive control for agile interactions between
  autonomous ground vehicles,
\newblock In \emph{Proceedings of the IEEE International Conference on Robotics
  and Automation (ICRA)}. \textbf{2018} 2403--2410.

\bibitem{ko2020shared}
S.~Ko, R.~Langari,
\newblock Shared control between human driver and machine based on game
  theoretical model predictive control framework,
\newblock In \emph{Proceedings of the IEEE/ASME International Conference on
  Advanced Intelligent Mechatronics (AIM)}. \textbf{2020} 649--654.

\bibitem{engwerda_algorithms_2007}
J.~Engwerda,
\newblock Algorithms for computing {Nash} equilibria in deterministic {LQ}
  games,
\newblock \emph{Computational Management Science} \textbf{2007}, \emph{4}, 2
  113.

\bibitem{vamvoudakis_multi-player_2011}
K.~G. Vamvoudakis, F.~L. Lewis,
\newblock Multi-player non-zero-sum games: {Online} adaptive learning solution
  of coupled {Hamilton}–{Jacobi} equations,
\newblock \emph{Automatica} \textbf{2011}, \emph{47}, 8 1556.

\bibitem{li_policy_2019}
X.~Li, Z.~Peng, L.~Liang, W.~Zha,
\newblock Policy iteration based {Q}-learning for linear nonzero-sum quadratic
  differential games,
\newblock \emph{Science China Information Sciences} \textbf{2019}, \emph{62}, 5
  52204.

\bibitem{zhou_nonlinear_2024}
Y.~Zhou, Y.~Li, A.~Sheng, G.~Qi,
\newblock A nonlinear multi-player non-zero-sum differential game with state
  constraints and input delay,
\newblock \emph{Nonlinear Dynamics} \textbf{2024}.

\bibitem{jiang_approximate_2020}
F.~Jiang, X.~Guo, X.~Zhang, Z.~Zhang, D.~Dong,
\newblock Approximate {Soft} {Policy} {Iteration} {Based} {Reinforcement}
  {Learning} for {Differential} {Games} with {Two} {Pursuers} versus {One}
  {Evader},
\newblock In \emph{2020 5th {International} {Conference} on {Advanced}
  {Robotics} and {Mechatronics} ({ICARM})}. IEEE, Shenzhen, China,
\newblock ISBN 978-1-72816-479-3, \textbf{2020} 471--476,
\newblock \urlprefix\url{https://ieeexplore.ieee.org/document/9195328/}.

\bibitem{abouheaf_policy_2019}
M.~I. Abouheaf, M.~S. Mahmoud, F.~L. Lewis,
\newblock Policy {Iteration} {Solution} for {Differential} {Games} with
  {Constrained} {Control} {Policies},
\newblock In \emph{2019 {American} {Control} {Conference} ({ACC})}. IEEE,
  Philadelphia, PA, USA,
\newblock ISBN 978-1-5386-7926-5, \textbf{2019} 4301--4306,
\newblock \urlprefix\url{https://ieeexplore.ieee.org/document/8814441/}.

\bibitem{mu_policy-iteration-based_2021}
C.~Mu, K.~Wang, C.~Sun,
\newblock Policy-{Iteration}-{Based} {Learning} for {Nonlinear} {Player} {Game}
  {Systems} {With} {Constrained} {Inputs},
\newblock \emph{IEEE Transactions on Systems, Man, and Cybernetics: Systems}
  \textbf{2021}, \emph{51}, 10 6488.

\bibitem{sani_realtime_2023}
M.~Sani, A.~Hably, B.~Robu, J.~Dumon,
\newblock Real‐time game‐theoretic model predictive control for
  differential game of target defense,
\newblock \emph{Asian Journal of Control} \textbf{2023}, \emph{25}, 5 3343.

\bibitem{di_lamperski_ddp_2019}
B.~Di, A.~Lamperski,
\newblock Newton's method and differential dynamic programming for
  unconstrained nonlinear dynamic games \textbf{2019}, preprint. Available as
  arXiv:1906.09097.

\bibitem{le_cleach_algames_2022}
S.~Le~Cleac'h, M.~Schwager, Z.~Manchester,
\newblock Algames: a fast augmented lagrangian solver for constrained dynamic
  games,
\newblock \emph{Autonomous Robots} \textbf{2022}, \emph{46}, 1.

\bibitem{saltari_nash_2022}
E.~Saltari, W.~Semmler, G.~Di~Bartolomeo,
\newblock A {Nash} {Equilibrium} for {Differential} {Games} with
  {Moving}-{Horizon} {Strategies},
\newblock \emph{Computational Economics} \textbf{2022}, \emph{60}, 3 1041.

\bibitem{benenati_linear-quadratic_2024}
E.~Benenati, S.~Grammatico,
\newblock Linear-{Quadratic} {Dynamic} {Games} as {Receding}-{Horizon}
  {Variational} {Inequalities}, \textbf{2024},
\newblock \urlprefix\url{http://arxiv.org/abs/2408.15703},
\newblock ArXiv:2408.15703 [cs, eess, math].

\bibitem{engwerda_lq_2005}
J.~Engwerda,
\newblock \emph{{LQ} {Dynamic} {Optimization} and {Differential} {Games}},
\newblock John Wiley \& Sons Ltd, The Atrium, Southern Gate, Chichester, West
  Sussex PO19 8SQ, England, \textbf{2005}.

\bibitem{chang_collaborative_2023}
C.~T. Chang, M.~Hebert, B.~Hayes,
\newblock Collaborative {Planning} and {Negotiation} in {Human}-{Robot}
  {Teams},
\newblock In \emph{Companion of the 2023 {ACM}/{IEEE} {International}
  {Conference} on {Human}-{Robot} {Interaction}}. ACM, Stockholm Sweden,
\newblock ISBN 978-1-4503-9970-8, \textbf{2023} 763--765,
\newblock \urlprefix\url{https://dl.acm.org/doi/10.1145/3568294.3579978}.

\bibitem{chiu_what_2024}
C.-Y. Chiu, J.~Li, M.~Bhatt, N.~Mehr,
\newblock To {What} {Extent} {Do} {Open}-{Loop} and {Feedback} {Nash}
  {Equilibria} {Diverge} in {General}-{Sum} {Linear} {Quadratic} {Dynamic}
  {Games}?,
\newblock \emph{IEEE Control Systems Letters} \textbf{2024}, 1--1.

\bibitem{williams_best_2018}
G.~Williams, B.~Goldfain, P.~Drews, J.~M. Rehg, E.~A. Theodorou,
\newblock Best {Response} {Model} {Predictive} {Control} for {Agile}
  {Interactions} {Between} {Autonomous} {Ground} {Vehicles},
\newblock In \emph{2018 {IEEE} {International} {Conference} on {Robotics} and
  {Automation} ({ICRA})}. IEEE, Brisbane, QLD,
\newblock ISBN 978-1-5386-3081-5, \textbf{2018} 2403--2410,
\newblock \urlprefix\url{https://ieeexplore.ieee.org/document/8462831/}.

\bibitem{johnson_approximate_2015}
M.~Johnson, R.~Kamalapurkar, S.~Bhasin, W.~E. Dixon,
\newblock Approximate \${N}\$ -{Player} {Nonzero}-{Sum} {Game} {Solution} for
  an {Uncertain} {Continuous} {Nonlinear} {System},
\newblock \emph{IEEE Transactions on Neural Networks and Learning Systems}
  \textbf{2015}, \emph{26}, 8 1645.

\bibitem{frihauf_nash_2012}
P.~Frihauf, M.~Krstic, T.~Basar,
\newblock Nash {Equilibrium} {Seeking} in {Noncooperative} {Games},
\newblock \emph{IEEE Transactions on Automatic Control} \textbf{2012},
  \emph{57}, 5 1192.

\bibitem{cappello_approximate_2020}
D.~Cappello, T.~Mylvaganam,
\newblock Approximate {Nash} {Equilibrium} {Solutions} of {Linear} {Quadratic}
  {Differential} {Games},
\newblock \emph{IFAC-PapersOnLine} \textbf{2020}, \emph{53}, 2 6685.

\bibitem{varga_upper_2024}
B.~Varga,
\newblock On the {Upper} {Bound} of {Near} {Potential} {Differential} {Games},
\newblock \emph{Results in Applied Mathematics} \textbf{2024}, \emph{22}
  100453.

\bibitem{mironchenko_input--state_2023}
A.~Mironchenko,
\newblock \emph{Input-to-{State} {Stability}: {Theory} and {Applications}},
\newblock Communications and {Control} {Engineering}. Springer International
  Publishing, Cham, \textbf{2023}.

\bibitem{handwerker2026iterated}
K.~Handwerker, B.~Varga,
\newblock Iterated-best-response-based model predictive control for
  shared-control in teleoperation,
\newblock In \emph{Proceedings of the 23rd IFAC World Congress}. IFAC,
  \textbf{2026} In press.

\bibitem{kagstrom_bounds_1977}
B.~Kågström,
\newblock Bounds and perturbation bounds for the matrix exponential,
\newblock \emph{BIT} \textbf{1977}, \emph{17}, 1 39.

\end{thebibliography}
\end{document}